\documentclass[9pt,draft]{amsart}
\usepackage{mathrsfs}
\usepackage{amssymb}
\usepackage{amsmath}
\usepackage{amsfonts}
\usepackage{amsfonts,enumerate}
\usepackage{pifont}
\usepackage{enumerate}
\usepackage{graphics}
\usepackage{verbatim}
\usepackage{amsthm}

\numberwithin{equation}{section}
\newtheorem{theorem}{Theorem}[section]

\newtheorem{definition}{Definition}[section]
\newtheorem{lemma}{Lemma}[section]

\newtheorem{remark}{Remark}[section]

\newcommand{\8}{\infty}
\newcommand{\el}{\ell}

\newcommand{\be}{\begin{eqnarray*}}
\newcommand{\ee}{\end{eqnarray*}}
\newcommand{\beq}{\begin{equation}}
\newcommand{\eeq}{\end{equation}}
\newcommand{\beqn}{\begin{equation*}}
\newcommand{\eeqn}{\end{equation*}}
\newcommand{\bsp}{\begin{split}}
\newcommand{\esp}{\end{split}}

\begin{document}

\title{Blowup alternative for Gross-Pitaevskii hierarchies}


\author{Zeqian Chen}

\address{Wuhan Institute of Physics and Mathematics, Chinese Academy of Sciences, West District 30, Xiao-Hong-Shan, Wuhan 430071, China}


\author{Chengjun He}

\address{Wuhan Institute of Physics and Mathematics, Chinese Academy of Sciences, West District 30, Xiao-Hong-Shan, Wuhan 430071, China}

\author{Chuangye Liu}

\address{School of Mathematics and Statistics, Central China Normal University, Luo-Yu Road 152, Wuhan 430079, China}



\date{}
\maketitle
\markboth{Z. Chen, C. He, and C. Liu}%
{Gross-Pitaevskii hierarchies}

\begin{abstract}
In this paper, we prove the blowup alternative for Gross-Pitaevskii hierarchies on $\mathbb{R}^n$ and give the associated lower bounds on the blowup rate. In particular, we prove that any solution of density operators to the focusing Gross-Pitaevskii hierarchy blow up in finite time for $n \ge 3$ if the energy per some $k$ particles in the initial condition is negative. All of these results hold without the assumption of factorized conditions for initial values as well as the admissible ones. Our analysis is based on use of a quasi-Banach space of sequences of marginal density matrices.
\end{abstract}


\section{Introduction}\label{Intro}

The time-dependent Gross-Pitaevskii (GP, in short) equation describes the dynamics of initially trapped Bose-Einstein condensates (see \cite{DGPS, ESY2007b} for detailed information). Precisely, in units where $\hbar =1$ and the mass of the bosons $m=1/2,$ the condensate wave function at time $t,$ $\varphi_t = \varphi_t (\mathbf{r}),$ $\mathbf{r} \in \mathbb{R}^3,$ satisfies the GP equation
\beq\label{eq:GPeq}
\mathrm{i} \partial_t \varphi_t = - \Delta_\mathbf{r} \varphi_t + \sigma |\varphi_t |^2 \varphi_t
\eeq
with the normalization $\int | \varphi_t (\mathbf{r}) |^2 d ^3 \mathbf{r} = 1,$ where $\sigma = 8 \pi N a$ is the coupling constant with $N$ being the number of particles and $a$ the scattering length of the interaction potential. This fact was rigorously proved by Erd\"{o}s, Schlein, and Yau \cite{ESY2006, ESY2007a, ESY2009, ESY2010}, starting from a many-body bosonic Schr\"{o}dinger equation with a short-scale repulsive interaction in the dilute limit.

As for the motivation of the present paper, we give a brief description of the main ingredients found in \cite{ESY2006, ESY2007a, ESY2009, ESY2010}. To this end, we assume that the ground state of initially trapped Bose gases exhibits a complete condensation (see \cite{LS2002, LSY2000} for detailed information). After instantaneously removing the trap, the evolution of $N$ bosons is given by the Hamiltonian
\beq\label{eq:N-Hamilton}
H_N = \sum^N_{j=1} (- \triangle_{\mathbf{r}_j}) + \sum^N_{i < j} V_N (\mathbf{r}_i - \mathbf{r}_j).
\eeq
The interaction potential is scaled as $V_N (\mathbf{r}) = N^2 V(N \mathbf{r}),$ where $V$ is a positive, spherically symmetric, compactly supported, smooth potential. Let $\psi_{N, t}$ be the solution to the Schr\"{o}dinger equation
\beq\label{eq:N-SchEq}
\mathrm{i} \partial_t \psi_{N, t} = H_N \psi_{N, t}\quad \text{with}\quad \psi_{N, 0} = \psi.
\eeq
For $k =1, \ldots, N,$ let $\gamma^{(k)}_{N, t}$ denote the $k$-particle marginal density of $\psi_{N, t}$ with a consistent normalization $\mathrm{T r} \gamma^{(k)}_{N, t} =1.$ The time evolution of these marginal density matrices is then governed by a hierarchy of $N$ coupled equation, commonly known as the Bogoliubov-Born-Green-Kirkwood-Yvon (BBGKY, in short) hierarchy
\beq\label{eq:BBGKY}
\mathrm{i} \partial_t \gamma^{(k)}_{N, t} = \sum^k_{j =1} [-\triangle_{\mathbf{r}_j}, \gamma^{(k)}_{N, t} ] + (N-k) \sum^k_{j =1} \mathrm{T r}_{k+1} [V_N (\mathbf{r}_j - \mathbf{r}_{k+1}), \gamma^{(k+1)}_{N, t} ],
\eeq
where $\mathrm{T r}_{k+1}$ denotes the partial trace over $\mathbf{r}_{k+1}.$ Furthermore, they show that $\{ \gamma^{(k)}_{N, t} \}_{N \ge k}$ has at least one $w^*$-limit point $\gamma^{(k)}_{\8, t}$ in the space of trace class operators, and thereby the time evolution of $\gamma^{(k)}_{\8, t}$ is described by the Gross-Pitaevskii hierarchy
\beq\label{eq:GPhierarchy}
\mathrm{i} \partial_t \gamma^{(k)}_{\8, t} = \sum^k_{j =1} [-\triangle_{\mathbf{r}_j}, \gamma^{(k)}_{\8, t} ] + \sigma B^{(k)} \gamma^{(k+1)}_{\8, t},
\eeq
with the collision operator $B^{(k)} = \sum^k_{j =1} B^{(k)}_j$ being defined according to
\beq\label{eq:CollisionOper}
B^{(k)}_j \gamma^{(k+1)} = \mathrm{T r}_{k+1} [\delta (\mathbf{r}_j - \mathbf{r}_{k+1}), \gamma^{(k+1)} ],\quad j=1,\ldots, k,
\eeq
and with $\sigma = \int V (\mathbf{r}) d^3 \mathbf{r}.$ Finally, they proves that for any factorized initial condition
\beq\label{eq:InitialStateProduct}
\gamma_{\8, 0}^{(k)} = | \varphi_0 \rangle \langle \varphi_0 |^{\otimes^k}, \quad k \ge 1, \quad \text{with}\quad \varphi_0 \in H^1 (\mathbb{R}^3),
\eeq
the solution of the GP hierarchy \eqref{eq:GPhierarchy} is unique and factorized: The family
\beq\label{eq:SolutionProduct}
\gamma_{\8, t}^{(k)} = | \varphi_t \rangle \langle \varphi_t |^{\otimes^k}, \quad k \ge 1, \quad \text{with}\quad \varphi_t \in H^1 (\mathbb{R}^3),
\eeq
is a solution to \eqref{eq:GPhierarchy} if and only if $\varphi_t$ is a solution to the GP equation \eqref{eq:GPeq}.

While the existence of factorized solutions can be easily obtained, the proof of uniqueness of solutions of the GP hierarchy is the most difficult part of this analysis and was originally obtained in \cite{ESY2007a} by use of highly sophisticated Feynman graph expansion methods inspired by quantum field theory. Their uniqueness result (Theorem 9.1 in \cite{ESY2007a} or Theorem 4 in \cite{ESY2007b}) reads as: Given a family of density operators $\Gamma = (\gamma^{(k)})_{k \ge 1}$ such that
\beq\label{eq:ESYnorm}
||| \gamma^{(k)} |||_k : = \mathrm{T r} \big [ | (1- \triangle_{\mathbf{r}_1})^{\frac{1}{2}} \cdots (1- \triangle_{\mathbf{r}_k})^{\frac{1}{2}} \gamma^{(k)} (1- \triangle_{\mathbf{r}_1})^{\frac{1}{2}} \cdots (1- \triangle_{\mathbf{r}_k})^{\frac{1}{2}} | \big ] \le C^k,\quad \forall k \ge 1,
\eeq
for some $C>0,$ then for any $T>0,$ there exists at most one solution of density operators $\Gamma (t) = (\gamma^{(k)} (t) )_{k \ge 1}$ to \eqref{eq:GPhierarchy} with $\Gamma (0) = \Gamma$ and such that $||| \gamma^{(k)} (t) |||_k \le C^k$ holds uniformly in $t \in [0, T).$

An alternative method for proving uniqueness has been subsequently proposed by Klainerman and Machedon in \cite{KM2008} based on use of space-time bounds on the density matrices and introduction of an elegant ``board game" argument. Roughly speaking, they proved the uniqueness of solutions to \eqref{eq:GPhierarchy} in the product space
$\dot{\mathcal{H}}^1 = \bigotimes^{\8}_{k=1} \dot{\mathrm{H}}^1_k$ of homogeneous Sobolev spaces $\dot{\mathrm{H}}^1_k = \dot{\mathrm{H}}^1 (\mathbb{R}^{3 k} \times \mathbb{R}^{3 k})$ for all $k \ge 1,$ under the assumption of a particular {\it a priori} space-time bound on the density matrices:
\beq\label{eq:KM-AprioriBoundCollisionOper}
\| B^{(k)}_j \gamma^{(k+1)} \|_{L^1_t \dot{\mathrm{H}}^1_k} \le C^k,\quad j=1,\ldots,k,
\eeq
with $C$ independent of $k$ (see Section \ref{Pre} for notions used here). Based on energy conservation, K. Kirkpatrick, B. Schlein, and G. Staffilani in \cite{KSS2011} proved recently that the inequality \eqref{eq:KM-AprioriBoundCollisionOper} is indeed satisfied in the case of $\mathbb{R}^2.$

Let $\mathfrak{H}^1$ be the space of all sequences $\Gamma = (\gamma^{(k)})_{k \ge 1}$ of trace class operators such that $\| \Gamma \|_{\mathfrak{H}^1} < \8,$ where 
\be
\| \Gamma \|_{\mathfrak{H}^1} : = \inf \left \{ \lambda>0:\; \sum_{k=1}^{\infty} \frac{1}{\lambda^k} |||\gamma^{(k)}|||_k \le 1 \right \}.
\ee
Note that $\Gamma = (\gamma^{(k)})_{k \ge 1}$ satisfies \eqref{eq:ESYnorm} if and only if $\Gamma \in \mathfrak{H}^1.$ Then the uniqueness theorem of Erd\"{o}s, Schlein, and Yau stated above can be reformulated as follows: {\it The initial problem for the GP hierarchy \eqref{eq:GPhierarchy} admits at most one solution of density operators in $\mathfrak{H}^1.$} This note also applies to the uniqueness theorem of Klainerman and Machedon \cite{KM2008} along with the condition \eqref{eq:KM-AprioriBoundCollisionOper}.

Motivated by this observation, the first named author in \cite{Chen} introduced a class of quasi-Banach spaces for studying the Cauchy problem of the GP hierarchy \eqref{eq:GPhierarchy}. Briefly speaking, a certain Sobolev type space $\mathcal{H}^{s}$ of sequences of trace class operators is defined as a quasi-Banach space consisting of all $\Gamma = ( \gamma^{(k)} )_{k \ge 1}\in \bigotimes_{k=1}^{\infty} \mathrm{H}^{s}_k$ such that
\be
\sum_{k=1}^{\infty} \frac{1}{\lambda^k} \|\gamma^{(k)}\|_{\mathrm{H}^{s}_k} < \infty \;\text{for some}\; \lambda>0,
\ee
where $\mathrm{H}^{s}_k = \mathrm{H}^{s} (\mathbb{R}^{ k n} \times \mathbb{R}^{k n}),$ equipped with the quasi-norm
\be
\big \| \Gamma \big \|_{\mathcal{H}^{s}} : = \inf \left \{ \lambda>0:\; \sum_{k=1}^{\infty} \frac{1}{\lambda^k} \|\gamma^{(k)}\|_{\mathrm{H}^{s}_k} \le 1 \right \}.
\ee
It was then proved in \cite{Chen} that the Cauchy problem for the GP hierarchy \eqref{eq:GPhierarchy} is locally well posed in $\mathcal{H}^s$ for $s > \max \{\frac{1}{2}, \frac{n-1}{2} \}.$ The precise statement of this that we will use later is presented in Section \ref{Pre}.

In this paper, we continue this line of investigation. First of all, we present some notations and preliminaries in Section \ref{Pre}. Then, in Section \ref{BlowupAlter} we will prove the blowup alternative for GP hierarchies on $\mathbb{R}^n$ and give the associated lower bounds on the blowup rate. The conservation of energy and Virial identities of solutions to the GP hierarchy are presented in Section \ref{EnergyVirialIdentity}. In Section \ref{BlowupFinite}, we prove that any solution of density operators to the focusing GP hierarchy blow up in finite time for $n \ge 3$ if the energy per some $k$ particles in the initial condition is negative. Finally, in Section \ref{GPQuintic}, we will extend all the results for the cubic GP hierarchy to the so-called quintic one.

All these results hold without the assumption of factorized conditions for initial values. Yet we do not require the so-called admissible condition proposed in \cite{CPT2010} (see Remark \ref{rk:AdmissibleCondition} below), and hence make some improvements of the associated results found there.

\section{Preliminaries }\label{Pre}

As follows, we usually denote by $x= (x^1, \ldots, x^n)$ a general variable in $\mathbb{R}^n$ and by $\mathbf{x}_k=(x_1,\ldots, x_k)$ a point in $\mathbb{R}^{k n} = ( \mathbb{R}^n )^k.$ For any $x , y \in \mathbb{R}^n$ we denote by $x \cdot y = \sum^n_{i=1} x^i y^i$ and $|x|^2 = x \cdot x.$ For any $\mathbf{x}_k, \mathbf{y}_k \in \mathbb{R}^{k n},$ we set $\langle \mathbf{x}_k, \mathbf{y}_k \rangle: = \sum^k_{j=1} x_j \cdot y_j$ and $| \mathbf{x}_k |^2 = \langle \mathbf{x}_k, \mathbf{x}_k \rangle.$

We consider the Cauchy problem (initial value problem) for the Gross-Pitaevskii infinite linear hierarchy of equations on $\mathbb{R}^n,$ of the form
\begin{equation}\label{eq:GPHierarchyEquaFunct}
\left \{ \begin{split}
& \big ( \mathrm{i} \partial_t + \triangle^{(k)} \big ) \gamma^{(k)}_t ({\bf x}_k;{\bf x}'_k) = \mu \big [ B^{(k)} (\gamma^{(k+1)}_t )\big ] ({\bf x}_k; {\bf x}'_k ),\\
& \gamma^{(k)}_{t=0}({\bf x}_k;{\bf x}'_k) = \gamma^{(k)}_{0}({\bf x}_k;{\bf x}'_k), \quad k =1, 2, \ldots, \end{split} \right.
\end{equation}
where $t \in \mathbb{R}, {\bf x}_k = (x_1, x_2, \ldots, x_k), {\bf x}'_k = (x'_1, x'_2, \ldots, x'_k) \in \mathbb{R}^{k n}, \mu= \pm 1.$ Here,
\be
\triangle^{(k)} = \sum^k_{j=1} ( \Delta_{x_j} - \Delta_{x'_j} )\quad \text{and}\quad B^{(k)} = \sum^k_{j=1} B^{(k)}_j,
\ee
where $\Delta_{x_j}$ refers to the usual Laplace operator with respect to the variables $x_j \in {\mathbb R}^n$ and the operators $B^{(k)}_j = B^{(k)}_{j,+ }- B^{(k)}_{j,-}$ are defined according to
\begin{equation*}
\begin{split}
[ B^{(k)}_{j, +}& (\gamma^{(k+1)})]  ({\bf x}_{k},{\bf x}'_{k})\\
= & \int \mathrm{d} x_{k+1} \mathrm{d}x'_{k+1}
\delta(x_{k+1}-x'_{k+1}) \delta(x_j- x_{k+1}) \gamma^{(k+1)}({\bf
x}_{k+1},{\bf x}'_{k+1}),
\end{split}
\end{equation*}
and
\begin{equation*}
\begin{split}
[B^{(k)}_{j, - }& (\gamma^{(k+1)}) ] ({\bf x}_{k},{\bf x}'_{k})\\
= & \int \mathrm{d} x_{k+1} \mathrm{d} x'_{k+1}
\delta(x_{k+1}-x'_{k+1} ) \delta(x'_j - x_{k+1})\gamma^{(k+1)}({\bf
x}_{k+1},{\bf x}'_{k+1}).
\end{split}
\end{equation*}
A sequence of functions $\Gamma (t) = ( \gamma^{(k)}_t )_{k \ge 1}$ is said to be a Gross-Pitaevskii (GP, in short) hierarchy, if they satisfy \eqref{eq:GPHierarchyEquaFunct} and are symmetric, in the sense that
\be
\gamma^{(k)}_t ({\bf x}_k, {\bf x}'_k) = \overline{\gamma^{(k)}_t ({\bf x}'_k, {\bf x}_k)}
\ee
and
\be
\gamma^{(k)}_t (x_1,\dotsc,x_k;x'_1,\dotsc,x'_k)= \gamma^{(k)}_t (x_{\sigma(1)},\dotsc,x_{\sigma(k)};x'_{\sigma(1)},\dotsc,x'_{\sigma(k)})
\ee
for any $\sigma \in \Pi_k$ ($\Pi_k$ denotes the set of permutations on $k$ elements).

Let $\varphi \in \mathrm{H}^1(\mathbb{R}^n),$ then one can easily verify that a particular GP hierarchy, i.e., a solution to \eqref{eq:GPHierarchyEquaFunct} with factorized initial datum
\be
\gamma^{(k)}_{t=0}({\bf x}_k; {\bf x}'_k) = \prod^k_{j=1} \varphi(x_j) \overline{\varphi(x'_j)},\quad k=1,2,\ldots,
\ee
is given by
\be
\gamma^{(k)}_t ({\bf x}_{k}; {\bf x}'_{k} ) = \prod^k_{j=1} \varphi_t (x_j) \overline{\varphi_t ( x'_j )},\quad  k=1,2,\ldots,
\ee
where $\varphi_t$ satisfies the cubic non-linear Schr\"odinger equation
\beq\label{eq:GPEqua}
\mathrm{i} \partial_t \varphi_t = -\Delta \varphi_t + \mu |\varphi_t|^2 \varphi_t,\quad \varphi_{t=0}=\varphi,
\eeq
which is {\it defocusing} if $\mu =1,$ and {\it focusing} if $\mu= -1.$ We refer to \cite{C2003} and references therein for the nonlinear Schr\"odinger equation.

In the following, unless otherwise specified, we always use $\gamma^{(k)}, \rho^{(k)}$ for denoting symmetric functions in $\mathbb{R}^{k n} \times \mathbb{R}^{k n}.$ For $k \geq 1$ and $s \in \mathbb{R},$ we denote by $\mathrm{H}^{s}_k = \mathrm{H}^{s} (\mathbb{R}^{ k n} \times \mathbb{R}^{k n})$ the space of measurable functions $\gamma^{(k)} = \gamma^{(k)} ( {\bf x}_k, {\bf x}'_k )$ in $L^2(\mathbb{R}^{k n} \times \mathbb{R}^{k n}),$ which are symmetric, such that
\beq\begin{split}\label{normH}
\| \gamma^{(k)} \|_{\mathrm{H}^{s}_k} : = \| S^{(k)}_{s} \gamma^{(k)} \|_{L^2(\mathbb{R}^{k n} \times \mathbb{R}^{k n})} < \8,
\end{split}\eeq
where
\be\begin{split}
S^{(k)}_{s} :  = \prod_{j=1}^k  (1 - \Delta_{x_j} )^{\frac{s}{2}} (1 - \Delta_{x'_j} )^{\frac{s}{2}},
\end{split}\ee
with the convention $S^{(k)} = S^{(k)}_1.$ Evidently, $\mathrm{H}^{s}_k$ is a Hilbert space with the inner product
\be\begin{split}
\langle \gamma^{(k)}, \rho^{(k)} \rangle = \big \langle S^{(k)}_{s} \gamma^{(k)},\; S^{(k)}_{s} \rho^{(k)} \big \rangle_{L^2(\mathbb{R}^{k n} \times \mathbb{R}^{k n})}.
\end{split}\ee
Moreover, the norm $\|\cdot\|_{\mathrm{H}^{s}_k}$ is invariance under the action of $e^{\mathrm{i} t \triangle^{(k)}},$ i.e.,
\be
\| e^{\mathrm{i} t \triangle^{(k)}} \gamma^{(k)} \|_{\mathrm{H}^{s}_k} = \| \gamma^{(k)} \|_{\mathrm{H}^{s}_k}
\ee
because $e^{\mathrm{i} t \triangle^{(k)}}$ commutates with $\Delta_{x_j}$ for any $j.$

\begin{definition}\label{df:SobolevSpace} {\rm (cf. \cite{Chen}, Definition 2.1)}
For $s \in \mathbb{R}$ we define
\be
\mathcal{H}^{s} = \left \{ ( \gamma^{(k)} )_{k \ge 1} \in \bigotimes_{k=1}^{\infty} \mathrm{H}^{s}_k :\; \sum_{k=1}^{\infty} \frac{1}{\lambda^k} \|\gamma^{(k)}\|_{\mathrm{H}^{s}_k} < \infty \;\text{for some}\; \lambda > 0 \right \},
\ee
equipped with the quasi-norm
\beq\label{eq:SobolevSpaceNorm}
\big \| ( \gamma^{(k)} )_{k \ge 1} \big \|_{\mathcal{H}^{s}} : = \inf \left \{ \lambda > 0:\; \sum_{k=1}^{\infty} \frac{1}{\lambda^k} \|\gamma^{(k)}\|_{\mathrm{H}^{s}_k} \le 1 \right \}.
\eeq
\end{definition}

\begin{remark}\label{rk:SobolevSpace}\rm
Note that $\| \cdot \|_{\mathcal{H}^{s}}$ is not actually a norm in the sense that it does not satisfy $\| \lambda \Gamma \|_{\mathcal{H}^{s}} = |\lambda| \| \Gamma \|_{\mathcal{H}^{s}}$ in general. However, it indeed satisfies the triangle inequality and hence, $\mathcal{H}^{s}$ is a quasi-Banach space.
\end{remark}

\begin{definition}\label{df:SpacetimeSobolevSpace} {\rm (cf. \cite{Chen}, Definition 2.3)}
For an interval $I \subset \mathbb{R},$ we define $L_{t\in I}^1\mathcal{H}^{s}$ to be the space of all strongly measurable functions $\Gamma(t)=\{\gamma^{(k)}_t\}_{k\geq 1}$ on $I$ with values in $\mathcal{H}^{s}$ such that
\begin{equation*}
\sum_{k=1}^{\infty} \frac{1}{\lambda^{k}}\int_I \|\gamma^{(k)}_t\|_{\mathrm{H}_k^{s}}dt<\infty\quad \text{for some} \ \lambda>0,
\end{equation*}
equipped with the quasi-norm
\begin{equation}\label{ndef2}
\|\Gamma(t)\|_{L_{t \in I}^1 \mathcal{H}^{s}}: = \inf \left \{ \lambda>0:  \sum_{k=1}^{\infty} \frac{1}{\lambda^{k}}\int_I \|\gamma^{(k)}_t\|_{\mathrm{H}_k^{s}} d t \leq 1 \right \}.
\end{equation}
\end{definition}

\begin{remark}\label{rk:normL}\rm
For any $\Gamma(t), \Phi(t) \in L_{t\in I}^1\mathcal{H}^{s}$ one can verify that
\begin{enumerate}[{\rm (1)}]

\item $\|\lambda \Gamma(t)\|_{L_{t \in I}^1\mathcal{H}^{s}}\leq\max\{1,|\lambda|\}\|\Gamma(t)\|_{L_{t\in I}^1\mathcal{H}^{s}}$ for any $\lambda \in \mathbb{C};$

\item $\|\Gamma(t) + \Phi(t) \|_{L_{t \in I}^1 \mathcal{H}^{s}} \leq \|\Gamma(t)\|_{L_{t \in I}^1\mathcal{H}^{s}}+\|\Phi(t)\|_{L_{t \in I}^1\mathcal{H}^{s}};$

\item $\|\Gamma(t)\|_{L_{t \in I}^1 \mathcal{H}^{s}}\leq \max\{1,|I| \} \|\Gamma(t)\|_{C(I, \mathcal{H}^{s})},$ where $|I|$ denotes the Lebesgue measure of $I;$

\item $\|\Gamma(t)\|_{L_{t \in I}^1 \mathcal{H}^{s}} = \|\Gamma(t)\|_{L_{t \in I_1}^1\mathcal{H}^{s}}+ \|\Gamma(t)\|_{L_{t \in I_2}^1\mathcal{H}^{s}}$ for subintervals $I_1, I_2 \subset I$ satisfying $I = I_1 \cup I_2$ and $I_1 \cap I_2 = \emptyset.$
\end{enumerate}
\end{remark}

Recall that, in integral formulation, \eqref{eq:GPHierarchyEquaFunct} can be written as
\beq\label{eq:GPHierarchyFunctIntEqua}
\gamma^{(k)}_t = e^{\mathrm{i} t \triangle^{(k)}} \gamma^{(k)}_0 + \int^{t}_{0} d s\;  e^{\mathrm{i} (t-s) \triangle^{(k)}} \tilde{B}^{(k)} \gamma^{(k+1)}_s,\; k=1,2,\ldots\;,
\eeq
whereafter $\tilde{B}^{(k)} = - \mathrm{i} \mu B^{(k)}.$ As noted in \cite{CP2010}, such a solution can be obtained by solving the following infinite linear hierarchy of integral equations
\beq\label{eq:GPStrongSolutionDuh-kth}\begin{split}
\tilde{B}^{(k)} \gamma^{(k+1)}_t & = \tilde{B}^{(k)} e^{\mathrm{i} t \triangle^{(k+1)}} \gamma^{(k+1)}_0 + \int^{t}_{0} d s\, \tilde{B}^{(k)} e^{\mathrm{i} (t-s) \triangle^{(k+1)}} \tilde{B}^{(k+1)} \gamma^{(k+2)}_s,
\end{split}\eeq
for any $k \ge 1.$ If we write
\be
\hat{\triangle} \Gamma : = \big ( \triangle^{(k)} \gamma^{(k)} \big )_{k \ge 1} \quad \text{and} \quad \hat{B} \Gamma := \big ( \tilde{B}^{(k)} \gamma^{(k+1)} \big )_{k \ge 1},
\ee
then \eqref{eq:GPHierarchyFunctIntEqua} and \eqref{eq:GPStrongSolutionDuh-kth} can be written as
\beq\label{eq:GPHierarchyFunctIntEquaGamma}
\Gamma (t) = e^{\mathrm{i} t \hat{\triangle}} \Gamma_0 + \int^{t}_{0} d s~  e^{\mathrm{i} (t-s) \hat{\triangle}} \hat{B} \Gamma (s)
\eeq
and
\beq\label{eq:GPStrongSolutionFunctGamma}\begin{split}
\hat{B} \Gamma (t) = \hat{B} e^{\mathrm{i} t \hat{\triangle}} \Gamma_0 + \int^{t}_{0} d s\, \hat{B} e^{\mathrm{i} (t-s) \hat{\triangle}} \hat{B} \Gamma (s),
\end{split}\eeq
respectively.

Let us make the notion of solution more precise.

\begin{definition}\label{df:StrongSolution}
A function $\Gamma (t) = ( \gamma^{(k)}_t )_{k \geq 1}: I \mapsto \mathcal{H}^{s}$ on a non-empty time interval $0 \in I \subset \mathbb{R}$ is said to be a local $(strong)$ solution
to the Gross-Pitaevskii hierarchy \eqref{eq:GPHierarchyEquaFunct} if it lies in the class $C (K, \mathcal{H}^{s})$ for all compact sets $K \subset I$ and obeys the Duhamel formula
\beq\label{eq:MildSolution}
\gamma^{(k)}_t = e^{\mathrm{i} t \triangle^{(k)}} \gamma^{(k)}_0 - \mathrm{i} \mu \int^{t}_{0} d s\; e^{\mathrm{i} (t-s) \triangle^{(k)}} B^{(k)} \gamma^{(k+1)}_s,\quad \forall t \in I,
\eeq
holds in $\mathrm{H}^{s}_k$ for every $k=1,2,\ldots.$
\end{definition}

We refer to the interval $I$ as the {\it lifespan} of $\Gamma (t)$ with the initial value $\Gamma (0) = ( \gamma^{(k)}_0 )_{k \geq 1}.$ We say that $\Gamma (t)$ is a {\it maximal-lifespan solution} on $I = (- T_{\mathrm{min}}, T_{\mathrm{max}})$ with $T_{\mathrm{max}} = T_{\mathrm{max}} (\Gamma_0) \in (0, \8]$ and $T_{\mathrm{min}} = T_{\mathrm{min}} (\Gamma_0) \in (0, \8]$ if the solution cannot be extended to any strictly larger interval. $(- T_{\mathrm{min}}, T_{\mathrm{max}})$ is said to be the {\it maximal lifespan} of $\Gamma (t)$ with the initial value $\Gamma (0) = ( \gamma^{(k)}_0 )_{k \geq 1}.$ We say that $\Gamma (t)$ is a {\it global} solution if $(- T_{\mathrm{min}}, T_{\mathrm{max}}) = \mathbb{R},$ i.e., $T_{\mathrm{max}} = T_{\mathrm{min}} = \8.$

\begin{remark}\label{rk:SobolevSpace}\rm
Let $\varphi \in \mathrm{H}^{s}(\mathbb{R}^n)$ and set for $k \ge 1,$
\be
\gamma^{(k)}_0 ({\bf x}_k; {\bf x}'_k) = \prod^k_{j=1} \varphi(x_j) \overline{\varphi(x'_j)}.
\ee
An immediate computation yields that
\be
\big \| ( \gamma^{(k)}_0 )_{k \ge 1} \big \|_{\mathcal{H}^{s}} = 2 \| \varphi \|^2_{\mathrm{H}^{s} (\mathbb{R}^n)}.
\ee
Thus, for $T>0,$ $\varphi_t \in C((-T,T), \mathrm{H}^{s})$ is a solution to \eqref{eq:GPEqua} with the initial value $\varphi_t |_{t=0} = \varphi$ if and only if \beq\label{eq:FactorGPHierarchy}
\Gamma (t) = ( \gamma^{(k)}_t )_{k \ge 1} \; \text{with}\; \gamma^{(k)}_t({\bf x}_k; {\bf x}'_k) = \prod^k_{j=1} \varphi_t(x_j) \overline{\varphi_t(x'_j)}
\eeq
is a solution to \eqref{eq:GPHierarchyEquaFunct} in $C((-T,T), \mathcal{H}^{s})$ with $\Gamma (0) = ( \gamma^{(k)}_0 )_{k \ge 1}.$
This yields that the Cauchy problem \eqref{eq:GPHierarchyEquaFunct} in $\mathcal{H}^{s}$ is equivalent to the one \eqref{eq:GPEqua}
in $\mathrm{H}^{s}$ for the case that initial conditions are factorized.
\end{remark}

The local well-posedness for the GP hierarchy \eqref{eq:GPHierarchyEquaFunct} in $\mathcal{H}^{s}$ was proved in \cite{Chen} for $s > \max \{1/2,\; (n-1)/2 \}.$ We state it as Theorems \ref{th:LocalWellposued-alpha>n/2} and \ref{th:LocalWellposued-alpha>(n-1)/2}.

\begin{theorem}\label{th:LocalWellposued-alpha>n/2} {\rm (cf. \cite{Chen}, Theorem 2.1)}
Assume that $n \geq 1$ and $s > \frac{n}{2}.$ The Cauchy problem \eqref{eq:GPHierarchyEquaFunct} is locally well posed. More precisely, there exists a constant $A_{n, s}>0$ depending only on $n$ and $s$ such that
\begin{enumerate}[{\rm (1)}]

\item For each $\Gamma_0 = ( \gamma^{(k)}_{0} )_{k \geq 1}\in \mathcal{H}^{s},$ let $I = [-T, T]$ with $T =\frac{A_{n, s}}{\| \Gamma_0 \|_{\mathcal{H}^{s}}}.$ Then there exists a solution $\Gamma (t) = ( \gamma^{(k)}_t )_{k \geq 1} \in C ( I, \mathcal{H}^{s})$ to the Gross-Pitaevskii hierarchy \eqref{eq:GPHierarchyEquaFunct} with the initial data $\Gamma_0$ such that
\begin{equation}\label{eq:SpacetimeEstimate-LocalWellposued1}
\| \Gamma (t) \|_{C( I,\mathcal{H}^{s})} \leq 2 \| \Gamma_0 \|_{\mathcal{H}^{s}}.
\end{equation}

\item Given $I_0 = [-T_0, T_0]$ with $T_0 > 0.$ If $\Gamma (t)$ and $\Gamma' (t)$ in $C ( I_0, \mathcal{H}^{s})$ are two solutions to \eqref{eq:GPHierarchyEquaFunct} with initial conditions $\Gamma_{t=0} = \Gamma_0$ and $\Gamma'_{t=0} = \Gamma'_0$ in $\mathcal{H}^{s}$ respectively, then
\begin{equation}\label{eq:SpacetimeEstimate-Stability}
\| \Gamma (t) - \Gamma' (t) \|_{C( I, \mathcal{H}^{s})} \leq 2 \| \Gamma_0 - \Gamma'_0  \|_{\mathcal{H}^{s}},
\end{equation}
with $I = [-T, T],$ where
\be
T = \min \left \{T_0, \frac{A_{n, s}}{\| \Gamma (t) - \Gamma' (t) \|_{C ( I_0; \mathcal{H}^{s})}} \right \}.
\ee
\end{enumerate}

\end{theorem}

\begin{remark}\label{rk:thalpha>n/2}\rm
Theorem \ref{th:LocalWellposued-alpha>n/2} shows that one has unconditional local wellposedness in $\mathcal{H}^{s}$ for the Cauchy problem of the GP hierarchy \eqref{eq:GPHierarchyEquaFunct} for any $s > n/2.$ This agrees with the case of the GP equation \eqref{eq:GPEqua} (see e.g. \cite{Tao2006}, Proposition 3.8). Using the classical argument (see e.g. the proof of Theorem 1.17 in \cite{Tao2006}), one can easily shows that given any $\Gamma_0 \in \mathcal{H}^{s},$ there exists a unique maximal interval of existence $I$ and a unique solution $\Gamma_t \in C (I, \mathcal{H}^{s}).$ Moreover, if $I$ has a finite endpoint $T,$ i.e., $T = T_{\max} < \8$ or $T = T_{\min} <\8,$ then the $\mathcal{H}^{s}$-norm of $\Gamma_t$ will go to infinity as $t \to T.$ (Indeed, we will give the corresponding blowup rate in Theorem \ref{th:BlowupRate-alpha>n/2} below.) Thus, the maximal lifespan of $\Gamma_t$ is necessarily open.
\end{remark}

\begin{theorem}\label{th:LocalWellposued-alpha>(n-1)/2} {\rm (cf. \cite{Chen}, Theorem 2.2)}
Assume that $n \ge 2$ and $s > \frac{n-1}{2}.$ Then, the Cauchy problem for the Gross-Pitaevskii hierarchy \eqref{eq:GPHierarchyEquaFunct} is locally well posed in $\mathcal{H}^{s}.$ More precisely, there exist an absolute constant $A>2$ and a constant $C=B_{n, s}>0$ depending only on $n$ and $s$ such that
\begin{enumerate}[{\rm (1)}]

\item For every $\Gamma_0 = ( \gamma^{(k)}_{0} )_{k \geq 1} \in \mathcal{H}^{s},$ let $I = [-T, T]$ with $T = \frac{B_{n, s}}{\| \Gamma_0 \|^2_{\mathcal{H}^{s}}}.$ Then there exists a solution $\Gamma (t) = ( \gamma^{(k)} (t) )_{k \geq 1} \in C ( I, \mathcal{H}^{s})$ to \eqref{eq:GPHierarchyEquaFunct} with the initial data $\Gamma (0) = \Gamma_0$ satisfying
\begin{equation}\label{eq:SpacetimeEstimate}
\| \hat{B} \Gamma (t) \|_{L^1_{t \in I} \mathcal{H}^{s}} \leq 4 A \| \Gamma_0 \|_{\mathcal{H}^{s}}.
\end{equation}

\item Given $I_0 = [-T_0, T_0]$ with $T_0 >0.$ If $\Gamma (t) \in C ( I_0, \mathcal{H}^{s})$ so that $\hat{B} \Gamma (t) \in L^1_{t \in I_0} \mathcal{H}^{s}$ is a solution to \eqref{eq:GPHierarchyEquaFunct} with the initial data $\Gamma (0) = \Gamma_0,$ then \eqref{eq:SpacetimeEstimate} holds as well for $I = [-T, T],$ where
\be
T = \min \left \{T_0, \; \frac{B_{n, s}}{ \| \hat{B} \Gamma (t) \|^2_{L^1_{t \in I_0} \mathcal{H}^{s}} + \| \Gamma_0 \|^2_{\mathcal{H}^{s}} } \right \}.
\ee

\item Given $I_0 = [-T_0, T_0]$ with $T_0 >0.$ If $\Gamma (t)$ and $\Gamma' (t)$ in $C ( I_0, \mathcal{H}^{s})$ with $\hat{B} \Gamma (t), \hat{B} \Gamma' (t) \in L^1_{t \in I_0}\mathcal{H}^{s}$ are two solutions to \eqref{eq:GPHierarchyEquaFunct} with initial conditions $\Gamma (0) = \Gamma_0$ and $\Gamma' (0) = \Gamma'_0$ in $\mathcal{H}^{s}$ respectively, then
\begin{equation}\label{eq:InitialValueContinuousDependence}
\| \Gamma (t) - \Gamma' (t) \|_{C( I, \mathcal{H}^{s})} \leq (1+ 4 A) \| \Gamma_0 - \Gamma'_0\|_{\mathcal{H}^{s}},
\end{equation}
with $I =[-T, T],$ where
\be
T = \min \left \{ T_0, \; \frac{B_{n, s}}{ \| \hat{B} [\Gamma (t) - \Gamma' (t) ]\|^2_{L^1_{t \in I_0} \mathcal{H}^{s}} + \| \Gamma_0 - \Gamma'_0 \|^2_{\mathcal{H}^{s}} } \right \}.
\ee

\end{enumerate}

In particular, the above results hold for $\mathcal{H}^1$ in the case $n=3.$
\end{theorem}

\begin{remark}\label{rk:thalpha>(n-1)/2}\rm
As shown in Theorem \ref{th:LocalWellposued-alpha>(n-1)/2}, for the case $n/2 \ge s > (n-1)/2$ we require a priori assumption $\hat{B} \Gamma (t) \in L^1_{t \in I} \mathcal{H}^{s}$ in both the stability and uniqueness parts, although we prove that for the existence part,
such a priori assumption is not required. At the time of this writing, the question remains open whether the
condition $\hat{B} \Gamma (t) \in L^1_{t \in I} \mathcal{H}^{s}$ is necessary for the uniqueness of solutions.
\end{remark}


\begin{remark}\label{rk:ChenPavlovicResult}\rm
Note that a similar but different solution space was previously introduced by Chen and Pavlovi\'{c} \cite{CP2010} for studying the initial problem of \eqref{eq:GPHierarchyEquaFunct}. Precisely, given $0<\xi <1,$ set
\be
\mathcal{H}^{s}_{\xi} = \left \{ \Gamma = ( \gamma^{(k)} )_{k \ge 1} \in \bigotimes_{k=1}^{\infty} \mathrm{H}^{s}_k :\; \| \Gamma \|_{\mathcal{H}^{s}_{\xi}}:= \sum_{k=1}^{\infty} \xi^k \|\gamma^{(k)}\|_{\mathrm{H}^{s}_k} < \infty \right \}.
\ee
Then, equipped with the norm $\| \cdot \|_{\mathcal{H}^{s}_{\xi}},$ $\mathcal{H}^{s}_{\xi}$ is a Banach space. The local well-posedness obtained in \cite{CP2010} states that for any initial data $\Gamma_0 \in \mathcal{H}^{s}_{\xi_1}$ with $\xi_1 > 0,$ there exists a unique solution $\Gamma (t) \in C([-T, T], \mathcal{H}^{s}_{\xi_2})$ for some $0 < \xi_2 < \xi_1$ and $T>0,$ under an additional assumption on $\hat{B} \Gamma (t)$ (see also \cite{CL2011} for some improvements in the case $s > \frac{n}{2}$). That is, there are two different parameters $\xi_1, \xi_2$ in their result. On the other hand, the norm $\| \varphi_t \|_{\mathrm{H}^{s}}$ is not compatible with $\| \Gamma (t) \|_{\mathcal{H}^{s}_{\xi}}$ for factorized hierarchies $\Gamma$ of the form \eqref{eq:FactorGPHierarchy}.

As shown in Theorems \ref{th:LocalWellposued-alpha>n/2} and \ref{th:LocalWellposued-alpha>(n-1)/2}, these two undesirable issues will be eliminated if the space $\mathcal{H}^{s}$
is involved instead of $\mathcal{H}^{s}_{\xi}.$ This shows that the space $\mathcal{H}^{s}$ with the
quantity \eqref{eq:SobolevSpaceNorm} seems more suitable for studying the Cauchy problem of the GP hierarchy \eqref{eq:GPHierarchyEquaFunct}.
\end{remark}

\section{Blowup alternative and blowup rate in finite time}\label{BlowupAlter}

In this section, we prove two results concerning the blowup alternative (e.g., see Definition 3.1.5 in \cite{C2003}) and the associated lower bounds on the blow-up rate of solutions to the GP hierarchy \eqref{eq:GPHierarchyEquaFunct}.

\begin{theorem}\label{th:BlowupRate-alpha>n/2}
Assume that $n\geq 1$ and $s > \frac{n}{2}.$ If $\Gamma(t)$ is a solution to the Gross-Pitaevskii hierarchy \eqref{eq:GPHierarchyEquaFunct} with initial condition $\Gamma(0)\in \mathcal{H}^{s}$ such that $T_{\mathrm{max}} < \8,$ then
$\lim_{t \nearrow T_{\mathrm{max}}} \| \Gamma (t) \|_{\mathcal{H}^{s}} = \8,$ and there exists a constant $A_{n,s}>0$ depending only on $n$ and $s$ such that the following lower bound on the blowup rate holds
\beq\label{eq:Belowuprate-alpha>n/2}\begin{split}
\|\Gamma(t)\|_{\mathcal{H}^{s}} & \geq \frac{A_{n,s}}{T_{\mathrm{max}}-t},\quad \forall \ 0<t<T_{\mathrm{max}}.
\end{split}\eeq

The similar results hold for $T_{\min}.$
\end{theorem}

\begin{proof}
Suppose $T_{\mathrm{max}} < \8.$ We only need to prove \eqref{eq:Belowuprate-alpha>n/2}. For any $0< t < T_{\mathrm{max}},$ let
\be
\mathcal{R}_{t}\gamma^{(k)}(\tau,{\bf x}_k;{\bf x}_k'):=\gamma^{(k)} \Big (t + \tau, {\bf x}_k;{\bf x}_k'\Big), \quad \tau\in \mathbb{R}.
\ee
We note that if $\Gamma(t)=\{\gamma^{(k)}(t)\}_{k\geq 1}$ solves the GP hierarchy \eqref{eq:GPHierarchyEquaFunct}, then $\mathcal{R}_{t}\Gamma(\tau) = ( \mathcal{R}_{t}\gamma^{(k)}(\tau) )_{k\geq 1}$ is also a solution of the GP hierarchy \eqref{eq:GPHierarchyEquaFunct}.

To prove \eqref{eq:Belowuprate-alpha>n/2}, we fix $t \in (0, T_{\mathrm{max}}).$ Clearly,
\begin{equation}\label{eq:S}
\| \mathcal{R}_{t}\gamma^{(k)}(\tau)\|_{\mathrm{H}^{s}_k} =\|\gamma^{(k)} (t+\tau)\|_{\mathrm{H}^{s}_k},
\end{equation}
which implies that
\begin{equation}\label{eq:Gamma}
\|\mathcal{R}_{t}\Gamma(\tau)\|_{{\mathcal H}^{s}} = \|\Gamma(t+\tau)\|_{{\mathcal H}^{s}}.
\end{equation}
Let $\tau=0$ in \eqref{eq:Gamma}, we obtain
\begin{equation}\label{eq:Gamma0}
\|\mathcal{R}_{t}\Gamma(0)\|_{{\mathcal H}^{s}} = \|\Gamma(t)\|_{{\mathcal H}^{s}} <\infty.
\end{equation}
Then, Theorem \ref{th:LocalWellposued-alpha>n/2} implies that there exists
$\tau^*= A_{n,s}/\|\mathcal{R}_{t}\Gamma(0)\|_{\mathcal{H}^{s}}$
such that
\begin{equation}\label{eq:RGamma}
\|\mathcal{R}_{t}\Gamma(\tau)\|_{{\mathcal H}^{s}}\leq
2\|\mathcal{R}_{t}\Gamma(0)\|_{\mathcal{H}^{s}}
\end{equation}
for any $0\leq \tau\leq \tau^*$. Therefore, combining \eqref{eq:Gamma}
 with \eqref{eq:RGamma} we get
\be\begin{split}
\|\Gamma(t+\tau^*)\|_{{\mathcal H}^{s}} =\|\mathcal{R}_{t}\Gamma(\tau^*)\|_{{\mathcal H}^{s}} \leq 2\|\mathcal{R}_{t}\Gamma(0)\|_{\mathcal{H}^{s}} = 2\|\Gamma(t)\|_{\mathcal{H}^{s}} <\infty.
\end{split}\ee
Hence, the upper maximal-lifespan time $T_{\mathrm{max}}$ is bounded from below by
\be
T_{\mathrm{max}} > t+\tau^*,
\ee
and thus
\be\begin{split}
T_{\mathrm{max}} - t & > \tau^* = \frac{A_{n,s}}{\|\mathcal{R}_{t}\Gamma(0)\|_{\mathcal{H}^{s}}} = \frac{A_{n,s}}{\|\Gamma(t)\|_{\mathcal{H}^{s}}}.
\end{split}
\ee
Consequently, we have
\be
\|\Gamma(t)\|_{\mathcal{H}^{s}} \geq \frac{A_{n,s}}{T_{\mathrm{max}} - t},\quad \forall 0< t < T_{\mathrm{max}},
\ee
as required.

If $T_{\mathrm{min}} < \8,$ we can proceed the same argument and omit the details.
\end{proof}

\begin{remark}\label{rk:BlowupLowerbound-alpha>n/2}\rm
Note that in the factorized case, $\| \Gamma \|_{\mathcal{H}^{s}} = 2 \| \varphi \|^2_{\mathrm{H}^{s}}$ if $\Gamma = (| \varphi \rangle \langle \varphi |^{\otimes^k} )_{k \ge 1}.$ Then the lower bound \eqref{eq:Belowuprate-alpha>n/2} on the blow-up rate coincides with the known one for solutions to the GP equation \eqref{eq:GPEqua} in the cases $n=1,2$ ( see e.g. \cite{C2003}).
\end{remark}

As shown in Theorem \ref{th:LocalWellposued-alpha>(n-1)/2}, we require an additional assumption that $\hat{B} \Gamma (t) \in L^1_{t \in [0,T]} \mathcal{H}^{s}$ for the uniqueness part in the case $s \le n/2.$ Then the lifespan $I$ of a solution $\Gamma (t)$ should be such that
\be
\|\Gamma(t)\|_{C(I,\mathcal{H}^{s})}+\|\hat{B}\Gamma(t)\|_{L_{t\in I}^1\mathcal{H}^{s}}<\infty.
\ee
Accordingly, we define the corresponding upper maximal-lifespan time $T_{\mathrm{max}}$ by
\be
T_{\max} = \sup \big \{ T>0: \|\Gamma(t)\|_{C([0, T],\mathcal{H}^{s})}+\|\hat{B}\Gamma(t)\|_{L_{t\in [0,T]}^1\mathcal{H}^{s}}<\infty \big \}.
\ee
Similarly, we can define $T_{\mathrm{min}}$ as well.

\begin{theorem}\label{th:BlowupRate-alpha>(n-1)/2}
Assume that $n\geq 2$ and $s > \frac{n-1}{2}.$ If $\Gamma(t)$ is a solution of the Gross-Pitaevskii hierarchy \eqref{eq:GPHierarchyEquaFunct} with initial condition $\Gamma(0)\in \mathcal{H}^{s}$ such that $T_{\mathrm{max}} < \8,$ then
$\lim_{t \nearrow T_{\mathrm{max}}} \| \Gamma (t) \|_{\mathcal{H}^{s}} = \8,$ and there exists a constant $B_{n,s}>0$ depending only on $n$ and $s$ such that the following lower bound on the blowup rate holds
\beq\label{eq:Belowuprate-alpha>(n-1)/2}\begin{split}
\|\Gamma(t)\|_{\mathcal{H}^{s}} & \geq \frac{B_{n,s}}{ ( T_{\mathrm{max}}-t )^{\frac{1}{2}}},\quad \forall 0<t<T_{\mathrm{max}}.
\end{split}\eeq
In particular, the above inequality holds for the $\mathcal{H}^1$-norm in the case $n=3.$

The similar results hold for $T_{\min}.$
\end{theorem}

\begin{proof}
It suffices to prove \eqref{eq:Belowuprate-alpha>(n-1)/2}. To this end, we let $T_{\mathrm{max}}< \8$ and fix $T \in (0, T_{\mathrm{max}})$ such that
\be
\|\Gamma(t)\|_{C([0,T],\mathcal{H}^{s})}+\|\hat{B}\Gamma(t)\|_{L_{t\in [0,T]}^1\mathcal{H}^{s}}<\infty.
\ee
Clearly, \eqref{eq:S} implies that for any $\tau >0$
\be
\|\hat{B}\Gamma(t)\|_{L_{t \in[T,T+\tau ]}^1{\mathcal H}^{s}}
= \|\mathcal{R}_{T}\hat{B}\Gamma(\tau)\|_{L_{\tau\in[0,\tau ]}^1{\mathcal H}^{s}}.
\ee
It follows from \eqref{eq:Gamma0} that
\be
\|\mathcal{R}_{T}\Gamma(0)\|_{{\mathcal H}^{s}}=
\|\Gamma(T)\|_{{\mathcal H}^{s}}<\infty.
\ee
Then, Theorem \ref{th:LocalWellposued-alpha>(n-1)/2} implies that there exists
$\tau^*=\frac{B_{n,s}}{\|\mathcal{R}_{T}\Gamma(0)\|^2_{\mathcal{H}^{s}}}$
such that
\begin{equation}\label{eq:BgammaL}
\|\hat{B}\mathcal{R}_{T}\Gamma(\tau)\|_{L_{\tau\in[0,\tau^*]}^1{\mathcal H}^{s}}\leq
4A\|\mathcal{R}_{T}\Gamma(0)\|_{\mathcal{H}^{s}}.
\end{equation}
Since
\begin{equation*}
\|\mathcal{R}_{T}\Gamma(\tau)\|_{C([0,\tau^*], \mathcal{H}^{s})}
\leq \| \mathcal{R}_{T}\Gamma(0)\|_{\mathcal{H}^{s}}+\|\hat{B}\mathcal{R}_{T}\Gamma(\tau)\|_{L_{\tau\in[0,\tau^*]}^1{\mathcal H}^{s}},
\end{equation*}
it follows from \eqref{eq:BgammaL} that
\begin{equation}\label{eq:CRT}
\|\mathcal{R}_{T}\Gamma(\tau)\|_{C([0,\tau^*]),{\mathcal H}^{s}}
\leq
(1+4A)\|\mathcal{R}_{T}\Gamma(0)\|_{\mathcal{H}^{s}}.
\end{equation}
Note that
\be
\|\hat{B}\Gamma(t)\|_{L_{t\in[0,T+\tau^*]}^1{\mathcal
H}^{s}}
=\|\hat{B}\Gamma(t)\|_{L_{t\in[0,T]}^1{\mathcal
H}^{s}}+\|\hat{B}\Gamma(t)\|_{L_{t\in[T,T+\tau^*]}^1{\mathcal
H}^{s}},
\ee
by \eqref{eq:BgammaL} we get
\begin{equation}\label{eq:BGamma}
\|\hat{B}\Gamma(t)\|_{L_{t\in[0,T+\tau^*]}^1{\mathcal
H}^{s}}
= \|\hat{B}\Gamma(t)\|_{L_{t\in[0,T]}^1{\mathcal
H}^{s}}+\|\mathcal{R}_{T}\hat{B}\Gamma(\tau)\|_{L_{\tau\in[0,\tau^*]}^1{\mathcal
H}^{s}}
\end{equation}
By \eqref{eq:Gamma} again, we obtain
\begin{equation}\label{eq:C0TGamma}
\|\Gamma(t)\|_{C([0,T+\tau^*],{\mathcal
H}^{s})}
\leq \|\Gamma(t)\|_{C([0,T],{\mathcal
H}^{s})}+\|\mathcal{R}_{T}\Gamma(\tau)\|_{C([0,\tau^*],{\mathcal
H}^{s})}
\end{equation}
Also, by the definitions of the operators $\hat{B}$ and $\mathcal{R}_T$ we have
\begin{equation}\label{eq:RBBR}
\mathcal{R}_{T}\hat{B}\Gamma(\tau)
=\hat{B}\mathcal{R}_{T}\Gamma(\tau).
\end{equation}
Then, it follows from \eqref{eq:BgammaL}, \eqref{eq:BGamma} and \eqref{eq:RBBR} that
\begin{equation}\label{eq:BGammaL0Ttau^*}
\begin{split}
\|\hat{B} \Gamma(t) \|_{L_{t\in[0,T+\tau^*]}^1{\mathcal H}^{s}}
&\leq \|\hat{B}\Gamma(t)\|_{L_{t\in[0,T]}^1{\mathcal H}^{s}}+4A\|\mathcal{R}_{T}\Gamma(0)\|_{\mathcal{H}^{s}}\\
& = \|\hat{B}\Gamma(t)\|_{L_{t\in[0,T]}^1{\mathcal H}^{s}}+4A\|\Gamma(T)\|_{{\mathcal H}^{s}}\\
&\leq 4A \Big ( \|\hat{B}\Gamma(t)\|_{L_{t\in[0,T]}^1{\mathcal H}^{s}}+\|\Gamma(t)\|_{C([0,T],{\mathcal H}^{s})} \Big ).
\end{split}
\end{equation}
On the other hand, it follows from \eqref{eq:CRT} and \eqref{eq:C0TGamma} that
\begin{equation}\label{eq:BGammaC0Ttau^*}
\begin{split}
\|\Gamma(t) \|_{C([0,T+\tau^*],{\mathcal H}^{s})}
& \leq \|\Gamma(t)\|_{C([0,T],{\mathcal
H}^{s})}+(1+4A)\|\Gamma(T)\|_{{\mathcal H}^{s}}\\
& \leq(2+4A) \|\Gamma(t)\|_{C([0,T],{\mathcal
H}^{s})}.
\end{split}
\end{equation}
Thus, from \eqref{eq:BGammaL0Ttau^*} and \eqref{eq:BGammaC0Ttau^*} we get
\begin{equation*}
\begin{split}
\|\Gamma(t) & \|_{C([0,T+\tau^*],{\mathcal
H}^{s})}+ \|\hat{B}\Gamma(t)\|_{L_{t\in[0,T+\tau^*]}^1{\mathcal
H}^{s}}\\
&\leq (2+8A) \Big ( \|\Gamma(t)\|_{C([0,T],{\mathcal H}^{s})}+\|\hat{B}\Gamma(t)\|_{L_{t\in[0,T]}^1{\mathcal
H}^{s}} \Big ) <\infty.
\end{split}
\end{equation*}
Hence, the upper maximal lifespan time $T_{\mathrm{max}}$ is bounded from below by
\begin{equation*}
T_{\mathrm{max}} > T + \tau^*
\end{equation*}
and thus
\begin{equation*}
\begin{split}
T_{\mathrm{max}}-T >\tau^* &=\frac{B_{n,s}}{\|\mathcal{R}_{T}\Gamma(0)\|^2_{\mathcal{H}^{s}}} = \frac{B_{n,s}}{\|\Gamma(T)\|^2_{\mathcal{H}^{s}}}.
\end{split}
\end{equation*}
Consequently, we have
\be\begin{split}
\|\Gamma(T)\|_{\mathcal{H}^{s}} \geq \frac{B'_{n,s}}{( T_{\max} -T )^{\frac{1}{2}}},
\end{split}\ee
with $B'_{n,s} = \sqrt{B_{n,s}},$ as required.

The proof for the case of the lower maximal-lifespan time $T_{\mathrm{min}}$ is similar and omitted.
\end{proof}

\begin{remark}\label{rk:BlopupLowerbound-alpha>(n-1)/2}\rm
In the factorized case, the lower bound \eqref{eq:Belowuprate-alpha>(n-1)/2} on the blow-up rate coincides with the
one for solutions to the GP equation \eqref{eq:GPEqua} in the case $n=3$ ( see e.g. \cite{C2003, CW1990}).
\end{remark}

\section{Conservation of energy and Virial identities}\label{EnergyVirialIdentity}

In this section, we will mostly work in Fourier (momentum) space. Following \cite{ESY2007a}, we use the convention that variables $p,q,r, p', q', r'$
always refer to $n$ dimensional Fourier variables, while $x, x', y, y', z, z'$ denote the position space variables. With this convention, the usual hat indicating the Fourier transform will be omitted. For example, for $k \geq 1$ the kernel of a bounded operator $A$ on $L^2 ( \mathbb{R}^{k n} )$ in position space is $K (
\mathbf{x}_k ; \mathbf{x}'_k),$ then in the momentum space it is given by the Fourier transform
\be
K ( \mathbf{q}_k ; \mathbf{q}'_k) = \big \langle K, e^{ - \mathrm{i} \langle \cdot, \mathbf{q}_k \rangle}e^{ \mathrm{i} \langle \cdot, \mathbf{q}'_k \rangle} \big \rangle = \int d \mathbf{x}_k d \mathbf{x}'_k K ( \mathbf{x}_k ; \mathbf{x}'_k ) e^{- \mathrm{i} \langle \mathbf{x}_k, \mathbf{q}_k \rangle}e^{ \mathrm{i} \langle \mathbf{x}'_k, \mathbf{q}'_k \rangle},
\ee
with the slight abuse of notation of omitting the hat on left hand side. Here,
\be
\langle \mathbf{x}_k, \mathbf{q}_k \rangle = \sum^k_{j=1} x_j \cdot q_j, \quad \forall \mathbf{x}_k = (x_1, \ldots , x_k), \mathbf{q}_k = (q_1, \ldots , q_k) \in \mathbb{R}^{k n}.
\ee
Thus, on kernels in the momentum space $B^{(k)}$ acts according to
\beq\label{eq:BOperatorFunctMomentum}
\begin{split}
\big [ B^{(k)}& \gamma^{(k+1)} \big ] ({\bf p}_k;{\bf p}'_k)\\
= & \sum^k_{j=1} \int d q_{k+1} d q'_{k+1}\\
&\;\; \times \Big \{ \gamma^{(k+1)}(p_1,\dotsc, p_j-q_{k+1}+q'_{k+1}, \dotsc, p_k,q_{k+1}; {\mathtt p}'_k, q'_{k+1})\\
& \quad - \gamma^{(k+1)} ({\bf p}_k, q_{k+1}; p'_1, \dotsc, p'_j + q_{k+1} - q'_{k+1}, \dotsc, p'_k, q'_{k+1} ) \Big \}\\
= & \sum^k_{j=1} \int d {\bf q}_{k+1} d {\bf q}'_{k+1} \Big [ \prod^k_{l \neq j} \delta(p_l - q_l) \delta (p'_l-q'_l) \Big ]\\
& \;\;\times \gamma^{(k+1)} ( {\bf q}_{k+1}; {\bf q}'_{k+1}) \Big \{ \delta(p'_j-q'_j) \delta \big ( p_j- [q_j+q_{k+1}-q'_{k+1} ] \big )\\
& \quad - \delta (p_j - q_j) \delta \big ( p'_j - [ q'_j+q'_{k+1}-q_{k+1} ] \big ) \Big \}.
\end{split}
\eeq

Following \cite{CPT2010}, we introduce
\beq\label{eq:energyformula1}
E_k(\Gamma ):= \sum_{j=1}^k \int d \mathbf{x}_k d \mathbf{x}'_k \delta (\mathbf{x}_k - \mathbf{x}'_k) \Big ( \frac{1}{2} \nabla_{x_j} \cdot \nabla_{x_j'}\gamma^{(k)} + \frac{\mu}{4} B^{(k)}_{j, +}\gamma^{(k+1)} \Big ).
\eeq
Note that
\beq\label{eq:identity}
\begin{split}
\int & dx dx' \delta(x-x')\nabla_x\cdot \nabla_x'A(x,x')\\
& = -\int dx dx'\delta(x-x')\Delta_xA(x,x')\\
& = -\int dx dx'\delta(x-x')\Delta_{x'}A(x,x')\\
\end{split}
\eeq
which was previously proved in \cite{CPT2010}. Then
\beq\label{eq:energyformula2}
E_k(\Gamma ):= \sum_{j=1}^k \int d \mathbf{x}_k d \mathbf{x}'_k \delta (\mathbf{x}_k - \mathbf{x}'_k) \Big ( - \frac{1}{2} \triangle_{x_j} \gamma^{(k)} + \frac{\mu}{4} B^{(k)}_{j, +}\gamma^{(k+1)} \Big ).
\eeq
Further, we note that each of the terms in the sum equals to the one obtained for $j=1,$ by the symmetry of $\gamma^{(k)} (\mathbf{x}_k ; \mathbf{x}'_k)$ and $\gamma^{(k+1)} (\mathbf{x}_{k+1} ; \mathbf{x}'_{k+1}).$ Indeed, since
\be
\nabla_{x_j} \cdot \nabla_{x_j'}\gamma^{(k)} (\mathbf{x}_k ; \mathbf{x}'_k) = \sum^n_{\el =1} \frac{\partial^2}{\partial x^{\el}_j \partial x'^{\el}_j} \gamma^{(k)} (\mathbf{x}_k ; \mathbf{x}'_k),
\ee
we have that\be\begin{split}
\int d \mathbf{x}_k d \mathbf{x}'_k \delta (\mathbf{x}_k - \mathbf{x}'_k ) \nabla_{x_j} \cdot \nabla_{x_j'}\gamma^{(k)} (\mathbf{x}_k ; \mathbf{x}'_k) = \int d \mathbf{x}_k \sum^n_{\el =1} \frac{\partial^2}{\partial x^{\el}_j \partial x^{\el}_j} \gamma^{(k)} (\mathbf{x}_k ; \mathbf{x}_k).
\end{split}\ee
Then, by the symmetry of $\gamma^{(k)} (\mathbf{x}_k ; \mathbf{x}'_k)$ with respect to the components of $\mathbf{x}_k$ and $\mathbf{x}'_k$ one has
\be
\int d \mathbf{x}_k d \mathbf{x}'_k \delta (\mathbf{x}_k - \mathbf{x}'_k ) \nabla_{x_j} \cdot \nabla_{x_j'}\gamma^{(k)} = \int d \mathbf{x}_k d \mathbf{x}'_k \delta (\mathbf{x}_k - \mathbf{x}'_k ) \nabla_{x_i} \cdot \nabla_{x_i'}\gamma^{(k)}
\ee
for all $1 \le i, j \le k.$ We note that the calculation for the interaction term was presented in \cite{CPT2010}. Thus,
\beq\label{eq:k-energyformula2}
E_k(\Gamma ) = k \int d \mathbf{x}_k d \mathbf{x}'_k \delta (\mathbf{x}_k - \mathbf{x}'_k) \Big ( \frac{1}{2} \nabla_{x_1} \cdot \nabla_{x_1'}\gamma^{(k)} + \frac{\mu}{4} B^{(k)}_{1, +} \gamma^{(k+1)} \Big ).
\eeq

\begin{remark}\label{rk:k-energyGPE}\rm
For factorized states $\Gamma = ( | \varphi\rangle \langle \varphi |^{\otimes^k} )_{k \ge 1}$ with $\varphi \in \mathrm{H}^1 (\mathbb{R}^n)$ one finds that
\be
E_k (\Gamma) = k \| \varphi \|^{2 (k-1)}_{L^2} \left ( \frac{1}{2} \| \nabla \varphi \|^2_{L^2}  + \frac{\mu}{4} \| \varphi \|^4_{L^4} \right ).
\ee
In this case, $E_1 (\Gamma)$ is the usual expression of the conserved energy for solutions of \eqref{eq:GPEqua}.
\end{remark}

In what follows, we prove energy conservation of $k$ particles for solutions $\Gamma (t)$ to the GP hierarchy \eqref{eq:GPHierarchyEquaFunct} for any $k.$

\begin{theorem}\label{th:k-energyconservation}
Assume that $\Gamma(t)$ is a solution of the Gross-Pitaevskii hierarchy \eqref{eq:GPHierarchyEquaFunct} with initial condition $\Gamma(0)\in \mathcal{H}^{s}.$ Then $E_k (\Gamma (t))$ is a conserved quantity, i.e.,
\beq\label{eq:k-energyconversation}
E_k (\Gamma (t)) = E_k (\Gamma (0)), \quad \forall t \in (-T_{\mathrm{min}}, T_{\mathrm{max}}).
\eeq
\end{theorem}

\begin{remark}\label{rk:AdmissibleCondition}\rm
This result was previously proved in \cite{CPT2010} under the assumption that $\Gamma (t)$ is admissible in the sense that
\be
\gamma^{(k)}_t (\mathbf{x}_k ; \mathbf{x}'_k ) = \int d x_{k+1} \gamma^{(k+1)}_t (\mathbf{x}_k, x_{k+1} ; \mathbf{x}'_k, x_{k+1} ),\quad \forall k \ge 1.
\ee
In that case, it can be shown that $E_k (\Gamma (t)) = k E_1 (\Gamma (t))$ and thus reduces to show that $E_1 (\Gamma (t))$ is a conserved quantity. On the other hand, we note that for a sequence of ``pure states" $\gamma^{(k)}  (\mathbf{x}_k ; \mathbf{x}'_k ) = \psi (\mathbf{x}_k) \overline{\psi(\mathbf{x}'_k )},$ the admissibility condition implies that $(\gamma^{(k)})_{k \ge 1}$ must be factorized, by Schmidt's decomposition theorem (see e.g. \cite{RS2003}). This indicates that the admissibility requirement seems a little restrictive.
\end{remark}

\begin{proof}
First of all, by \eqref{eq:k-energyformula2} we have
\be
E_k(\Gamma(t) ) = k \int d \mathbf{x}_k d \mathbf{x}'_k \delta (\mathbf{x}_k - \mathbf{x}'_k) \Big ( \frac{1}{2} \nabla_{x_1} \cdot \nabla_{x_1'}\gamma_t^{(k)} + \frac{\mu}{4} B^{(k)}_{1, +}\gamma_t^{(k+1)} \Big ).
\ee
Then, by \eqref{eq:GPHierarchyEquaFunct} we have
\be
\mathrm{i} \partial_t E_k(\Gamma(t))=k[(I)+(II)+(III)+(IV)]
\ee
where
\be\begin{split}
(I): & = -\frac{1}{2}\int d \mathbf{x}_k d \mathbf{x}'_k \delta (\mathbf{x}_k - \mathbf{x}'_k) \nabla_{x_1} \cdot \nabla_{x_1'}\Delta^{(k)}\gamma_t^{(k)},\\
(II): & = \frac{\mu}{2}\int d \mathbf{x}_k d \mathbf{x}'_k \delta (\mathbf{x}_k - \mathbf{x}'_k) \nabla_{x_1} \cdot \nabla_{x_1'}B^{(k)}\gamma_t^{(k+1)},\\
(III): & = -\frac{\mu}{4}\int d \mathbf{x}_k d \mathbf{x}'_k \delta (\mathbf{x}_k - \mathbf{x}'_k) B^{(k)}_{1, +} \Delta^{(k+1)}\gamma_t^{(k+1)},\\
(IV): & = \frac{1}{4}\int d \mathbf{x}_k d \mathbf{x}'_k \delta (\mathbf{x}_k - \mathbf{x}'_k)B^{(k)}_{1, +}B^{(k+1)}\gamma_t^{(k+2)}.
\end{split}\ee
In order to prove $\partial_tE_k(\Gamma(t))=0,$ we shall prove that $(I)=(IV)=0$ and $(II)+(III)=0.$

For term $(I)$, we note that
\be
\begin{split}
-2(I) = & \int d \mathbf{x}_k d \mathbf{x}'_k \delta (\mathbf{x}_k - \mathbf{x}'_k) \nabla_{x_1} \cdot \nabla_{x_1'}\Delta^{(k)}\gamma_t^{(k)}\\
=&\sum_{j=1}^k\int d \mathbf{x}_k d \mathbf{x}'_k \delta (\mathbf{x}_k - \mathbf{x}'_k) (\Delta_{x_j}-\Delta_{x'_j})\nabla_{x_1} \cdot \nabla_{x_1'}\gamma_t^{(k)}\\
=&\sum_{j=1}^k\int \prod_{l\neq j}\delta ({x}_l - {x}'_l) d \mathbf{x}_k d \mathbf{x}'_k \delta ({x}_j - {x}'_j) (\Delta_{x_j}-\Delta_{x'_j}) \nabla_{x_1} \cdot \nabla_{x_1'}\gamma_t^{(k)}\\
=& \sum_{j=1}^k \int \prod_{l\neq j}\delta ({x}_l - {x}'_l) d \mathbf{x}_k d \mathbf{x}'_k \delta ({x}_j - {x}'_j)\\
& \quad \times \int e^{\mathrm{i} (\langle \mathbf{x}_k, \mathbf{p}_k \rangle - \langle \mathbf{x}_k', \mathbf{p}_k' \rangle)} d \mathbf{p} d \mathbf{p}' ( |p_j|^2 - |p_j'|^2 )( p_1 \cdot p_1') \gamma^{(k)} ( \mathbf{p}; \mathbf{p}').
\end{split}
\ee
Since
\be\begin{split}
\int  d x_j d x'_j & d p_j d p'_j \delta ({x}_j - {x}'_j) e^{\mathrm{i} (x_j \cdot p_j -x'_j \cdot p'_j )} ( |p_j|^2 - |p_j'|^2 ) ( p_1 \cdot p_1') \gamma^{(k)} ( \mathbf{p}; \mathbf{p}')\\
=& \int d p_j d p'_j \delta (p_j - p'_j) ( |p_j|^2 - |p_j'|^2 ) ( p_1 \cdot p_1') \gamma^{(k)} ( \mathbf{p}; \mathbf{p}') =0,
\end{split}\ee
for any $j,$ we proves $(I)=0.$

For $(IV)$, we have
\be
\begin{split}
4(IV) = & \int d \mathbf{x}_k d \mathbf{x}'_k \delta (\mathbf{x}_k - \mathbf{x}'_k) B^{(k)}_{1, +} B^{(k+1)}\gamma_t^{(k+2)}\\
=& \sum_{j=1}^{k+1} \int d \mathbf{x}_k d \mathbf{x}'_k \delta (\mathbf{x}_k - \mathbf{x}'_k) B^{(k)}_{1,+} \Big ( \gamma_t^{(k+2)}(\mathbf{x}_{k+1},x_j;\mathbf{x}'_{k+1},x_j) \\
& \quad - \gamma_t^{(k+2)}(\mathbf{x}_{k+1},x'_j;\mathbf{x}'_{k+1},x'_j) \Big )\\
= & \sum_{j=1}^{k}\int d \mathbf{x}_k d \mathbf{x}'_k \delta (\mathbf{x}_k - \mathbf{x}'_k) \Big ( \gamma_t^{(k+2)}(\mathbf{x}_{k},x_1,x_j;\mathbf{x}'_{k},x_1,x_j)\\
& \quad - \gamma_t^{(k+2)}(\mathbf{x}_{k},x_1,x'_j;\mathbf{x}'_{k},x_1,x'_j) \Big )\\
= & \sum_{j=1}^{k}\int d \mathbf{x}_k \Big[\gamma_t^{(k+2)}(\mathbf{x}_{k},x_1,x_j;\mathbf{x}_{k},x_1,x_j) -\gamma_t^{(k+2)}(\mathbf{x}_{k},x_1,x_j;\mathbf{x}_{k},x_1,x_j)\Big]\\
= & 0.
\end{split}
\ee
This proves $(IV)=0.$

For $(II)$, we have
\be
\begin{split}
(II) = & \frac{\mu}{2}\int d \mathbf{x}_k d \mathbf{x}'_k \delta (\mathbf{x}_k - \mathbf{x}'_k) \nabla_{x_1} \cdot \nabla_{x_1'}B^{(k)}\gamma_t^{(k+1)}\\
= & \frac{\mu}{2}\sum_{j=1}^{k}\int d \mathbf{x}_k d \mathbf{x}'_k \delta (\mathbf{x}_k - \mathbf{x}'_k) \nabla_{x_1} \cdot \nabla_{x_1'}[ B^{(k)}_{j,+} - B^{(k)}_{j,-} ] \gamma_t^{(k+1)}\\
= & \frac{\mu}{2}\int d \mathbf{x}_k d \mathbf{x}'_k \delta (\mathbf{x}_k - \mathbf{x}'_k) \nabla_{x_1} \cdot \nabla_{x_1'} [ B^{(k)}_{1,+} - B^{(k)}_{1,-} ] \gamma_t^{(k+1)}\\
& + \frac{\mu}{2}\sum_{j=2}^{k} \int d \mathbf{x}_k d \mathbf{x}'_k \delta (\mathbf{x}_k - \mathbf{x}'_k) \nabla_{x_1} \cdot \nabla_{x_1'} [ B^{(k)}_{j, +} - B^{(k)}_{j,-} ] \gamma_t^{(k+1)}\\
= & \frac{\mu}{2}\int d \mathbf{x}_k d \mathbf{x}'_k \delta (\mathbf{x}_k - \mathbf{x}'_k) \nabla_{x_1} \cdot \nabla_{x_1'} [ B^{(k)}_{1,+} - B^{(k)}_{1,-} ] \gamma_t^{(k+1)} + \frac{\mu}{2}\sum_{j=2}^{k} \int d \mathbf{x}_k d \mathbf{x}'_k \delta (\mathbf{x}_k - \mathbf{x}'_k)\\
& \quad \times \Big [ \big ( \Delta_{x_1} \gamma_t^{(k+1)} \big )( \mathbf{x}_k,x_j; \mathbf{x}'_k,x_j) - \big ( \Delta_{x_1} \gamma_t^{(k+1)} \big ) (\mathbf{x}_k,x'_j;\mathbf{x}'_k,x'_j) \Big ].\\
\end{split}
\ee
Note that
\be
\begin{split}
\int d \mathbf{x}_k d \mathbf{x}'_k \delta (\mathbf{x}_k - \mathbf{x}'_k) \Big [ \big ( \Delta_{x_1}\gamma_t^{(k+1)} \big )(\mathbf{x}_k,x_j;\mathbf{x}'_k,x_j)- \big (\Delta_{x_1}\gamma_t^{(k+1)} \big )(\mathbf{x}_k,x'_j;\mathbf{x}'_k,x'_j) \Big ] =0.
\end{split}
\ee
We get
\be
\begin{split}
(II) = & \frac{\mu}{2}\int d \mathbf{x}_k d \mathbf{x}'_k \delta (\mathbf{x}_k - \mathbf{x}'_k) \nabla_{x_1} \cdot \nabla_{x_1'}(B_{1,k+1}^+-B_{1,k+1}^-)\gamma_t^{(k+1)}\\
= & - \frac{\mu}{2}\int d \mathbf{x}_k d \mathbf{x}'_k \delta (\mathbf{x}_k - \mathbf{x}'_k) \Big [ \big (\Delta_{x'_1}\gamma_t^{(k+1)} \big ) (\mathbf{x}_k,x_1;\mathbf{x}'_k,x_1) - \big ( \Delta_{x_1}\gamma_t^{(k+1)} \big ) ( \mathbf{x}_k,x'_1;\mathbf{x}'_k,x'_1) \Big ]\\
\end{split}
\ee
where we have used the identity \eqref{eq:identity}.

Now we turn to the term $(III).$ By symmetry of $\gamma_t^{(k+1)}$ in $\mathbf{x}_k$ and $\mathbf{x}'_k,$ we have
\be
\begin{split}
(I I I) = &-\frac{\mu}{4}\int d \mathbf{x}_k d \mathbf{x}'_k \delta (\mathbf{x}_k - \mathbf{x}'_k) B^{(k)}_{1, +} \Delta^{(k+1)} \gamma_t^{(k+1)}\\
= & - \frac{\mu}{4} \sum_{j=1}^{k+1} \int d \mathbf{x}_k d \mathbf{x}'_k \delta (\mathbf{x}_k - \mathbf{x}'_k) B^{(k)}_{1, +} (\Delta_{x_j} - \Delta_{x'_j} ) \gamma_t^{(k+1)}\\
= &  - \frac{\mu}{4} \int d \mathbf{x}_k d \mathbf{x}'_k \delta (\mathbf{x}_k - \mathbf{x}'_k) B^{(k)}_{1, +} (\Delta_{x_1} - \Delta_{x'_1} ) \gamma_t^{(k+1)}\\
& \quad - \frac{\mu}{4}\sum_{j=2}^{k+1}\int d \mathbf{x}_k d \mathbf{x}'_k \delta (\mathbf{x}_k - \mathbf{x}'_k) B^{(k)}_{1, +} (\Delta_{x_j} - \Delta_{x'_j} ) \gamma_t^{(k+1)}\\
= & - \frac{\mu}{2}\int d \mathbf{x}_k d \mathbf{x}'_k \delta (\mathbf{x}_k - \mathbf{x}'_k) B^{(k)}_{1, +} (\Delta_{x_1}-\Delta_{x'_1})\gamma_t^{(k+1)}\\
= & - \frac{\mu}{2}\int d \mathbf{x}_k d \mathbf{x}'_k \delta (\mathbf{x}_k - \mathbf{x}'_k)\\
& \quad \times \Big [ \big ( \Delta_{x_1}\gamma_t^{(k+1)} \big ) (\mathbf{x}_k,x'_1;\mathbf{x}'_k,x'_1) - \big ( \Delta_{x'_1}\gamma_t^{(k+1)} \big ) (\mathbf{x}_k,x_1;\mathbf{x}'_k,x_1) \Big ].\\
\end{split}
\ee
Thus, $(II)+(III)=0.$

In summary, we have $\partial_tE(\Gamma(t))=0.$ Therefore, $E(\Gamma(t))=E(\Gamma(0))$ is a conserved quantity.
\end{proof}

Now we turn to Virial type identities for the GP hierarchy \eqref{eq:GPHierarchyEquaFunct}. This is necessary for the application of Glassey's argument for blowup in finite time of solutions to the focusing ($\mu=-1$) GP hierarchy \eqref{eq:GPHierarchyEquaFunct}.

Given a solution $\Gamma (t)$ to \eqref{eq:GPHierarchyEquaFunct}, we define for any $k \ge 1,$
\be
V_k (\Gamma (t)): = \mathrm{T r} \big [ | \mathbf{x}_k |^2 \gamma^{(k)} (t) \big ] = \sum^k_{j=1} \int d \mathbf{x}_k |x_j|^2 \gamma^{(k)} (t, \mathbf{x}_k; \mathbf{x}_k).
\ee
By the symmetry of $\gamma^{(k)},$ we have
\be
V_k (\Gamma (t)) = k \int d \mathbf{x}_k |x_1|^2 \gamma^{(k)} (t, \mathbf{x}_k; \mathbf{x}_k).
\ee
Then, the following Virial type identity for the GP hierarchy \eqref{eq:GPHierarchyEquaFunct} holds.

\begin{theorem}\label{th:VirialIdentity}
Assume that $\Gamma(t)= ( \gamma^{(k)}_t )_{k \geq 1}$ solves the Gross-Pitaevskii hierarchy \eqref{eq:GPHierarchyEquaFunct}. Then for any $k \ge 1,$
\beq\begin{split}\label{eq:VirialIdentity}
\partial^2_t V_k (\Gamma (t)) & = 8 k \int d \mathbf{p}_k |p_1|^2 \gamma^{(k)} (\mathbf{p}_k; \mathbf{p}_k) + 2 nk \mu \int d \mathbf{x}_k \gamma^{(k+1)} (\mathbf{x}_k, x_1; \mathbf{x}_k, x_1).
\end{split}\eeq
\end{theorem}

\begin{remark}\label{rk:VirialIdentity}\rm
The identity \eqref{eq:VirialIdentity} was previously proved in \cite{CPT2010} under the assumption that $\Gamma (t)$ is admissible (see Remark \ref{rk:AdmissibleCondition} for detailed information).
\end{remark}

\begin{proof}
We write
\be
\gamma^{(k)}({\bf x}_k;{\bf x}_k')= \int d{\bf p}_kd{\bf p}_k'e^{\mathrm{i} \langle{\bf p}_k,{\bf x}_k\rangle - \mathrm{i} \langle{\bf p}_k',{\bf x}_k'\rangle}\gamma^{(k)}({\bf p}_k;{\bf p}_k')
\ee
and define
\be
\rho({\bf
x}_k):= \gamma^{(k)}({\bf x}_k;{\bf x}_k) = \int d{\bf p}_kd{\bf p}_k'  e^{ \mathrm{i} \langle{\bf p}_k-{\bf p}_k',{\bf x}_k\rangle}\gamma^{(k)}({\bf p}_k;{\bf p}_k').
\ee
Then,
\begin{equation}\label{eq:partialtrho}
\begin{split}
\partial_t\rho({\bf
x}_k) = & \int d{\bf p}_kd{\bf p}_k' e^{ \mathrm{i} \langle{\bf p}_k-{\bf p}_k',{\bf x}_k\rangle}\partial_t\gamma^{(k)}({\bf p}_k;{\bf p}_k')\\
= & \mathrm{i} \int d{\bf p}_kd{\bf p}_k' e^{ \mathrm{i} \langle{\bf p}_k-{\bf p}_k',{\bf x}_k\rangle} ( \Delta^{(k)}\gamma^{(k)}) ({\bf p}_k;{\bf p}_k')\\
& \quad - \mu \mathrm{i} \int d{\bf p}_kd{\bf p}_k' e^{ \mathrm{i} \langle{\bf p}_k-{\bf p}_k',{\bf x}_k\rangle} ( B^{(k)} \gamma^{(k+1)})({\bf p}_k;{\bf p}_k').
\end{split}
\end{equation}
We first have
\begin{equation}\label{eq:laplacian}
\begin{split}
\mathrm{i} \int d{\bf p}_kd{\bf p}_k' & e^{ \mathrm{i} \langle{\bf p}_k-{\bf p}_k',{\bf x}_k\rangle} \Delta^{(k)} \gamma^{(k)}({\bf p}_k;{\bf p}_k')\\
& = - \mathrm{i} \int d{\bf p}_kd{\bf p}_k' e^{ \mathrm{i} \langle{\bf p}_k-{\bf p}_k',{\bf x}_k\rangle}(|{\bf p}_k|^2-|{\bf p}_k'|^2)\gamma^{(k)}({\bf p}_k;{\bf p}_k')\\
&= - \mathrm{i} \int d{\bf p}_kd{\bf p}_k' e^{ \mathrm{i} \langle{\bf p}_k-{\bf p}_k',{\bf x}_k\rangle}\langle{\bf p}_k+{\bf p}_k',{\bf p}_k-{\bf p}_k'\rangle\gamma^{(k)}({\bf p}_k;{\bf p}_k')\\
&= - \nabla_{{\bf x}_k}\cdot\int d{\bf p}_kd{\bf p}_k' e^{ \mathrm{i} \langle{\bf p}_k-{\bf p}_k',{\bf x}_k\rangle}({\bf p}_k+{\bf p}_k')\gamma^{(k)}({\bf p}_k;{\bf p}_k').
\end{split}
\end{equation}

On the other hand, for $j=1,2,\ldots,k,$
\be
\begin{split}
B^{(k)}_{j,+}\gamma^{(k+1)}({\bf x}_k;{\bf x}_k') = \int d{\bf q}_{k+1}d{\bf q}_{k+1}' e^{\mathrm{i} (\langle {\bf q}_k,{\bf x}_k\rangle+q_{k+1}\cdot x_j-\langle {\bf q}_k',{\bf x}_k'\rangle-q_{k+1}'\cdot x_j)}\gamma^{(k+1)}({\bf q}_{k+1};{\bf q}_{k+1}').
\end{split}
\ee
Thus,
\be
\begin{split}
B^{(k)}_{j,+} &  \gamma^{(k+1)} ({\bf p}_k;{\bf p}_k')\\
= &  \int d{\bf x}_kd{\bf x}_k' e^{ - \mathrm{i} \langle{\bf p}_k,{\bf x}_k\rangle + \mathrm{i} \langle{\bf p}_k',{\bf x}_k'\rangle}B^{(k)}_{j,+}\gamma^{(k+1)}({\bf x}_k;{\bf x}_k')\\
= & \int d{\bf x}_kd{\bf x}_k'd{\bf q}_{k+1}d{\bf q}_{k+1}' e^{ - \mathrm{i} \langle{\bf p}_k,{\bf x}_k\rangle + \mathrm{i} \langle{\bf p}_k',{\bf x}_k'\rangle} e^{ \mathrm{i} (\langle {\bf q}_k,{\bf x}_k\rangle+q_{k+1}\cdot x_j-\langle {\bf q}_k',{\bf x}_k'\rangle-q_{k+1}'\cdot x_j)}\gamma^{(k+1)}({\bf q}_{k+1};{\bf q}_{k+1}')\\
= & \int d{\bf q}_{k+1}d{\bf q}_{k+1}'\delta({\bf q}_k'-{\bf p}_k')\delta(q_{k+1}+q_j-q_{k+1}'-p_j)\prod_{l\neq j}^k\delta(p_l-q_l)
\gamma^{(k+1)}({\bf q}_{k+1};{\bf q}_{k+1}')\\
= & \int dq_{k+1}dq_{k+1}'\gamma^{(k+1)}(p_1,\ldots, p_{j-1}, q_{k+1}'-q_{k+1}+p_j, p_{j +1}, \ldots,p_k,q_{k+1}; {\bf p}_k',q_{k+1}')
\end{split}
\ee
Similarly, one has
\be
\begin{split}
B^{(k)}_{j,-}& \gamma^{(k+1)}({\bf p}_k;{\bf p}_k')\\
= & \int d{\bf x}_kd{\bf x}_k' e^{- \mathrm{i} \langle{\bf p}_k,{\bf x}_k\rangle + \mathrm{i} \langle{\bf p}_k',{\bf x}_k'\rangle}B^{(k)}_{j,-}\gamma^{(k+1)}({\bf x}_k;{\bf x}_k')\\
= & \int d{\bf x}_kd{\bf x}_k'd{\bf q}_{k+1}d{\bf q}_{k+1}' e^{- \mathrm{i} \langle{\bf p}_k,{\bf x}_k\rangle + \mathrm{i} \langle{\bf p}_k',{\bf x}_k'\rangle}e^{ \mathrm{i} (\langle {\bf q}_k,{\bf x}_k\rangle+q_{k+1}\cdot x_j'-\langle {\bf q}_k',{\bf x}_k'\rangle-q_{k+1}'\cdot x_j')}\gamma^{(k+1)}({\bf q}_{k+1};{\bf q}_{k+1}')\\
= & \int d {\bf q}_{k+1}d{\bf q}_{k+1}'\delta({\bf q}_k-{\bf p}_k)\delta(q_{k+1}-q_j'-q_{k+1}'+p_j')\prod_{l\neq j}^k\delta(p_l'-q_l')
\gamma^{(k+1)}({\bf q}_{k+1};{\bf q}_{k+1}')\\
= & \int dq_{k+1}dq_{k+1}'\gamma^{(k+1)}({\bf p}_k, q_{k+1}; p_1',\ldots, p_{j-1}', q_{k+1}-q_{k+1}'+p_j', p_{j+1}', \ldots, p_k', q_{k+1}').
\end{split}
\ee
Therefore,
\be
\begin{split}
\int &  d {\bf p}_k d{\bf p}_k' e^{ \mathrm{i} \langle {\bf p}_k-{\bf p}_k', {\bf x}_k \rangle} ( B^{(k)}_j \gamma^{(k+1)} ) ({\bf p}_k; {\bf p}_k')\\
& = \int d {\bf p}_k d{\bf p}_k' e^{ \mathrm{i} \langle{\bf p}_k-{\bf p}_k', {\bf x}_k \rangle} \big [ (B^{(k)}_{j,+} \gamma^{(k+1)} ) ({\bf p}_k; {\bf p}_k') - ( B^{(k)}_{j,-} \gamma^{(k+1)} ) ({\bf p}_k; {\bf p}_k') \big ]\\
& = \int dq_{k+1}dq_{k+1}'d{\bf p}_kd{\bf p}_k'e^{ \mathrm{i} \langle{\bf p}_k-{\bf p}_k',{\bf x}_k\rangle}\gamma^{(k+1)}(p_1,\ldots,q_{k+1}'-q_{k+1}+p_j,\ldots,p_k,q_{k+1};{\bf p}_k',q_{k+1}')\\
& \quad - \int dq_{k+1}dq_{k+1}'d{\bf p}_kd{\bf p}_k'e^{ \mathrm{i} \langle{\bf p}_k-{\bf p}_k',{\bf x}_k\rangle}\gamma^{(k+1)}({\bf p}_k,q_{k+1};p_1',\ldots,q_{k+1}-q_{k+1}'+p_j',\ldots,p_k',q_{k+1}')\\
&=0.
\end{split}
\ee
where the last equality is obtained by applying the change of variables $p_j \rightarrow q_{k+1}' - q_{k+1}+p_j$ and $p'_j \rightarrow q'_{k+1}-q_{k+1}+p'_j$ in the second term of the second equality so that the difference $p_j-p_j'$ remains unchanged. Since
\be
B^{(k)}=\sum_{j=1}^kB^{(k)}_j=\sum_{j=1}^k(B^{(k)}_{j,+}-B^{(k)}_{j,-}),
\ee
we conclude that
\beq\label{eq:Bkgamma}
\int d{\bf p}_kd{\bf p}_k' e^{ \mathrm{i} \langle{\bf p}_k-{\bf p}_k',{\bf x}_k\rangle} ( B^{(k)} \gamma^{(k+1)})({\bf p}_k;{\bf p}_k') = 0.
\eeq

Therefore, by combining \eqref{eq:partialtrho},\eqref{eq:laplacian} and \eqref{eq:Bkgamma} we have
\beq\label{eq:DensityFlow}
\partial_t \rho+\nabla_{{\bf x}_k}\cdot P=0,
\eeq
where
\be
P:=\int d{\bf p}_kd{\bf p}_k' e^{ \mathrm{i} \langle{\bf p}_k-{\bf p}_k',{\bf x}_k\rangle}({\bf p}_k+{\bf p}_k')\gamma^{(k)}({\bf p}_k;{\bf p}_k').
\ee

Now, we define
\be
M:=\int d{\bf x}_k\langle{\bf x}_k,P\rangle.
\ee
The time derivative is given by
\beq\label{eq:MorawetzAction}
\partial_t M=\int d{\bf x}_k\langle{\bf x}_k,\partial_tP\rangle= I_M + II_M,
\eeq
where
\be
I_M = \mathrm{i} \int d {\bf x}_k d {\bf p}_k d {\bf p}_k'e^{ \mathrm{i} \langle{\bf p}_k-{\bf p}_k',{\bf x}_k \rangle} \langle{\bf x}_k, {\bf p}_k+{\bf p}_k' \rangle ( \Delta^{(k)}\gamma^{(k)} ) ({\bf p}_k;{\bf p}_k'),
\ee
and
\be
II_M = -\mu \mathrm{i} \int d{\bf x}_k d{\bf p}_kd{\bf p}_k' e^{ \mathrm{i} \langle{\bf p}_k-{\bf p}_k',{\bf x}_k\rangle}\langle{\bf x}_k,{\bf p}_k+{\bf p}_k'\rangle ( B^{(k)}\gamma^{(k+1)} ) ({\bf p}_k;{\bf p}_k').
\ee
For the term $I_M$ we have
\be
\begin{split}
I_M & = - \mathrm{i} \int d{\bf x}_kd{\bf p}_kd{\bf p}_k'e^{ \mathrm{i} \langle{\bf p}_k-{\bf p}_k',{\bf x}_k\rangle}\langle{\bf x}_k,{\bf p}_k+{\bf p}_k'\rangle(|{\bf p}_k|^2-|{\bf p}_k'|^2)\gamma^{(k)}({\bf p}_k;{\bf p}_k')\\
& = - \mathrm{i} \int d{\bf x}_kd{\bf p}_kd{\bf p}_k'e^{ \mathrm{i} \langle{\bf p}_k-{\bf p}_k',{\bf x}_k\rangle}\langle{\bf x}_k,{\bf p}_k+{\bf p}_k'\rangle\langle{\bf p}_k+{\bf p}_k',{\bf p}_k-{\bf p}_k'\rangle\gamma^{(k)}({\bf p}_k;{\bf p}_k')\\
& = - \mathrm{i} \int d{\bf x}_kd{\bf p}_kd{\bf p}_k'e^{ \mathrm{i} \langle{\bf p}_k-{\bf p}_k',{\bf x}_k\rangle}\langle
  ({\bf p}_k+{\bf p}_k')({\bf p}_k+{\bf p}_k')^T {\bf x}_k,{\bf p}_k-{\bf p}_k'\rangle\gamma^{(k)}({\bf p}_k;{\bf p}_k')\\
& = -\int d{\bf p}_kd{\bf p}_k'\gamma^{(k)}({\bf p}_k;{\bf p}_k')\int d{\bf x}_k \langle
  ({\bf p}_k+{\bf p}_k')({\bf p}_k+{\bf p}_k')^T{\bf x}_k, \nabla_{{\bf x}_k}e^{ \mathrm{i} \langle{\bf p}_k-{\bf p}_k',{\bf x}_k\rangle}\rangle\\
& = \int d{\bf p}_kd{\bf p}_k'\gamma^{(k)}({\bf p}_k;{\bf p}_k')\mathrm{T r}
({\bf p}_k+{\bf p}_k')({\bf p}_k+{\bf p}_k')^T\int d{\bf x}_k e^{ \mathrm{i} \langle{\bf p}_k-{\bf p}_k',{\bf x}_k\rangle}\\
& = \int d{\bf p}_kd{\bf p}_k'\gamma^{(k)}({\bf p}_k;{\bf p}_k')|{\bf p}_k+{\bf p}_k'|^2\delta({\bf p}_k-{\bf p}_k')\\
& = 4 \int d{\bf p}_k|{\bf p}_k|^2\gamma^{(k)}({\bf p}_k;{\bf p}_k'),
\end{split}
\ee
where  ${\bf x}_k,{\bf p}_k,{\bf p}_k'$ are all considered as $n k \times 1$ matrices, and $A^T$ denotes the transpose of a matrix $A.$

Next, we determine the term $II_M.$ To this end, for each $ j=1,2,\ldots,k$ we have
\be
\begin{split}
\int d & {\bf p}_k d{\bf p}_k' e^{ \mathrm{i} \langle{\bf p}_k-{\bf p}_k',{\bf x}_k\rangle}({\bf p}_k+{\bf p}_k') ( B^{(k)}_j\gamma^{(k+1)}) ({\bf p}_k;{\bf p}_k')\\
= & \int d{\bf p}_kd{\bf p}_k' e^{ \mathrm{i} \langle{\bf p}_k-{\bf p}_k',{\bf x}_k\rangle}({\bf p}_k+{\bf p}_k') \big [ ( B^{(k)}_{j,+}\gamma^{(k+1)}) ({\bf p}_k;{\bf p}_k')- ( B^{(k)}_{j,-}\gamma^{(k+1)}) ({\bf p}_k;{\bf p}_k' \big ]\\
= & \int d{\bf p}_kd{\bf p}_k'dq_{k+1}dq_{k+1}' e^{ \mathrm{i} \langle{\bf p}_k-{\bf p}_k',{\bf x}_k\rangle}({\bf p}_k+{\bf p}_k')\\
& \times \big [ \gamma^{(k+1)}(p_1,\ldots,q_{k+1}'-q_{k+1}+p_j,\ldots,p_k,q_{k+1};{\bf p}_k',q_{k+1}')\\
& \quad - \gamma^{(k+1)}({\bf p}_k,q_{k+1};p_1',\ldots,q_{k+1}-q_{k+1}'+p_j',\ldots,p_k',q_{k+1}') \big ]\\
\end{split}
\ee
In the last term, we apply the change of variables $p_j \rightarrow p_j - q_{k+1} + q_{k+1}'$ and $p'_j \rightarrow p_j' - q_{k+1}+q_{k+1}'$ so that the difference $p_j-p_j'$ remains unchanged. Then the above integral equals
\be
\begin{split}
\int d & {\bf p}_k d{\bf p}_k'dq_{k+1}dq_{k+1}' e^{ \mathrm{i} \langle{\bf p}_k-{\bf p}_k',{\bf x}_k\rangle}({\bf p}_k+{\bf p}_k')\\
& \quad \times \gamma^{(k+1)}(p_1,\ldots,q_{k+1}'-q_{k+1}+p_j,\ldots,p_k,q_{k+1};{\bf p}_k',q_{k+1}')\\
& - \int d{\bf p}_kd{\bf p}_k'dq_{k+1}dq_{k+1}' e^{ \mathrm{i} \langle{\bf p}_k-{\bf p}_k',{\bf x}_k\rangle}(p_1+p_1',\ldots,p_j+p_j'-2q_{k+1}+2q_{k+1}',\ldots,p_k+p_k')\\
& \quad \times\gamma^{(k+1)}(p_1,\ldots,q_{k+1}'-q_{k+1}+p_j,\ldots,p_k,q_{k+1};{\bf p}_k',q_{k+1}')\\
= & \int d{\bf p}_kd{\bf p}_k'dq_{k+1}dq_{k+1}' e^{ \mathrm{i} \langle{\bf p}_k-{\bf p}_k',{\bf x}_k\rangle}2(0,\ldots,q_{k+1}-q_{k+1}',\ldots,0)\\
& \quad \times \gamma^{(k+1)}(p_1,\ldots,q_{k+1}'-q_{k+1}+p_j,\ldots,p_k,q_{k+1};{\bf p}_k',q_{k+1}')\\
\end{split}
\ee
The contribution of this term to the integral $\int d{\bf x}_k \langle {\bf x}_k,\partial_tP\rangle$ is given by
\be
\begin{split}
-\mu \mathrm{i} \int d & {\bf x}_kd{\bf p}_kd{\bf p}_k'dq_{k+1}dq_{k+1}' e^{ \mathrm{i} \langle{\bf p}_k-{\bf p}_k',{\bf x}_k\rangle}2x_j\cdot (q_{k+1}-q_{k+1}')\\
&\quad \times \gamma^{(k+1)}(p_1,\ldots,q_{k+1}'-q_{k+1}+p_j,\ldots,p_k,q_{k+1};{\bf p}_k',q_{k+1}').
\end{split}
\ee
Now, we apply Fourier transform again. This integral equals
\begin{equation}\label{eq:2n+2x.nablax}
\begin{split}
-\mu \mathrm{i} \int d & {\bf x}_kd{\bf p}_kd{\bf p}_k'dq_{k+1}dq_{k+1}'d{\bf y}_{k+1}d{\bf y}_{k+1}'e^{ \mathrm{i} \langle{\bf p}_k-{\bf p}_k',{\bf x}_k\rangle}2x_j\cdot (q_{k+1}-q_{k+1}')\\
& \quad \times e^{- \mathrm{i} \langle {\bf p}_k,{\bf y}_k\rangle- \mathrm{i} (q_{k+1}'-q_{k+1})\cdot y_j- \mathrm{i} q_{k+1}\cdot y_{k+1} + \mathrm{i} \langle {\bf p}_k',{\bf y}_k'\rangle + \mathrm{i} q_{k+1}'\cdot y_{k+1}'} \gamma^{(k+1)}({\bf y}_{k+1};{\bf y}_{k+1}')\\
& = - \mu \mathrm{i} \int d{\bf x}_kd{\bf y}_{k+1}d{\bf y}_{k+1}'\gamma^{(k+1)}({\bf y}_{k+1};{\bf y}_{k+1}')
\int d{\bf p}_kd{\bf p}_k'dq_{k+1}dq_{k+1}'\\
& \quad \times e^{ \mathrm{i} \langle{\bf p}_k,{\bf x}_k -{\bf y}_k\rangle - \mathrm{i} \langle{\bf p}_k',{\bf x}_k -{\bf y}_k'\rangle}2x_j\cdot (q_{k+1}-q_{k+1}')e^{ \mathrm{i} q_{k+1}\cdot(y_j-y_{k+1}) - \mathrm{i} q_{k+1}'\cdot(y_j-y_{k+1}')}\\
&=-\mu\int d{\bf x}_kd{\bf y}_{k+1}d{\bf y}_{k+1}'\gamma^{(k+1)}({\bf y}_{k+1};{\bf y}_{k+1}')\delta({\bf x}_k -{\bf y}_k)\delta({\bf x}_k -{\bf y}_k')\\
& \quad \times \int dq_{k+1}dq_{k+1}'2x_j\cdot\nabla_{y_j}e^{ \mathrm{i} q_{k+1}\cdot(y_j-y_{k+1}) - \mathrm{i} q_{k+1}'\cdot(y_j-y_{k+1}')}\\
&= - \mu \int d {\bf x}_k d {\bf y}_{k+1} d y_{k+1}'\gamma^{(k+1)}({\bf y}_{k+1};{\bf x}_k,y_{k+1}')\delta({\bf x}_k -{\bf y}_k)\\
& \quad \times 2x_j \cdot \nabla_{y_j} \delta(y_j-y_{k+1})\delta(y_j-y_{k+1}')\\
& = - \mu \int d{\bf x}_kd y_{k+1}d y_{k+1}'\gamma^{(k+1)}({\bf x}_k,y_{k+1};{\bf x}_k,y_{k+1}')\\
& \quad \times 2x_j \cdot \nabla_{x_j} \delta(x_j-y_{k+1}) \delta(y_{k+1}-y_{k+1}')\\
&= - \mu \int d {\bf x}_kd y_{k+1}\gamma^{(k+1)}({\bf x}_k,y_{k+1};{\bf x}_k,y_{k+1}) 2x_j\cdot\nabla_{x_j}\delta(x_j-y_{k+1})\\
&=\mu\int d{\bf x}_kd y_{k+1}\delta(x_j-y_{k+1})(2n+2x_j\cdot\nabla_{x_j})\gamma^{(k+1)}({\bf x}_k,y_{k+1};{\bf x}_k,y_{k+1}). \\
\end{split}
\end{equation}
Note that
\begin{equation}\label{eq:x.nablax}
\begin{split}
\int & d{\bf x}_kx_j\cdot\nabla_{x_j}\gamma^{(k+1)}({\bf x}_k,x_j;{\bf x}_k,x_j)\\
= & \int d{\bf x}_kd y_{k+1}\delta(x_j-y_{k+1})(x_j\cdot\nabla_{x_j}+y_{k+1}\cdot\nabla_{y_{k+1}})\gamma^{(k+1)}({\bf x}_k,y_{k+1};{\bf x}_k,y_{k+1})\\
= & \int d{\bf x}_kd y_{k+1}\delta(x_j-y_{k+1}) \big [ x_j\cdot\nabla_{x_j}\gamma^{(k+1)}({\bf x}_k,y_{k+1};{\bf x}_k,y_{k+1})\\
& \quad +y_{k+1}\cdot\nabla_{y_{k+1}}\gamma^{(k+1)}(x_1,\ldots,y_{k+1},\ldots,x_k,x_j;x_1,\ldots,y_{k+1},\ldots,x_k,x_j) \big ]\\
= & \int d{\bf x}_kd y_{k+1}\delta(x_j-y_{k+1}) 2x_j \cdot \nabla_{x_j}\gamma^{(k+1)}({\bf x}_k,y_{k+1};{\bf x}_k,y_{k+1}) \\
\end{split}
\end{equation}
where we used the symmetry of $\gamma^{(k+1)},$ and renamed the variables in the last term. Also,
\begin{equation}\label{eq:-n}
\int d{\bf x}_kx_j\cdot\nabla_{x_j}\gamma^{(k+1)}({\bf x}_k,x_j;{\bf x}_k,x_j)=
-n\int d{\bf x}_k\gamma^{(k+1)}({\bf x}_k,x_j;{\bf x}_k,x_j)
\end{equation}
from integrating by parts.

Therefore, combing \eqref{eq:2n+2x.nablax},\eqref{eq:x.nablax} and \eqref{eq:-n} yields
\be
\begin{split}
II_M= & \sum_{j=1}^k\mu\int d{\bf x}_kd y_{k+1}\delta(x_j-y_{k+1})(2n+2x_j\cdot\nabla_{x_j})\gamma^{(k+1)}({\bf x}_k,y_{k+1};{\bf x}_k,y_{k+1})\\
= & \sum_{j=1}^k\mu\int d{\bf x}_kd y_{k+1}\delta(x_j-y_{k+1})(2n-n)\gamma^{(k+1)}({\bf x}_k,y_{k+1};{\bf x}_k,y_{k+1})\\
= & n k \mu\int d{\bf x}_k\gamma^{(k+1)}({\bf x}_k,x_1;{\bf x}_k,x_1).
\end{split}
\ee
Finally, we combine \eqref{eq:DensityFlow} and \eqref{eq:MorawetzAction} to conclude that
\be
\begin{split}
\partial_t^2 \int & d {\bf x}_k|{\bf x}_k|^2\gamma^{(k)}(t,{\bf x}_k;{\bf x}_k)\\ \
& = 2 \int d{\bf x}_k\langle {\bf x}_k,\partial_tP\rangle\\
&=8\int d{\bf p}_k|{\bf p}_k|^2\gamma^{(k)}({\bf p}_k;{\bf p}_k)+2nk\mu\int d{\bf x}_k\gamma^{(k+1)}({\bf x}_k,x_1;{\bf x}_k,x_1)\\
&=8k\int d{\bf p}_k| p_1|^2\gamma^{(k)}({\bf p}_k;{\bf p}_k)+2nk\mu\int d{\bf x}_k\gamma^{(k+1)}({\bf x}_k,x_1;{\bf x}_k,x_1).
\end{split}
\ee
This completes the proof of Theorem \ref{th:VirialIdentity}.
\end{proof}

\section{Blowup of solutions to the focusing GP hierarchy}\label{BlowupFinite}

In this section, using conservation of energy and Virial type identities obtained in the previous section, we prove a result on blowup of solutions to focusing GP hierarchies in finite time.


Recall that $\mathfrak{H}^1$ is the space of all sequences $\Gamma = (\gamma^{(k)})$ of trace class operators satisfying $\| \Gamma \|_{\mathfrak{H}^1} < \8,$ where
\be
\| \Gamma \|_{\mathfrak{H}^1} = \inf \left \{ \lambda>0:\; \sum_{k=1}^{\infty} \frac{1}{\lambda^k} |||\gamma^{(k)}|||_k \le 1 \right \},
\ee
and $||| \gamma^{(k)} |||_k = \mathrm{T r} \big [ | S^{(k)} \gamma^{(k)} | \big ]$ for any $k \ge 1.$

\begin{theorem}\label{th:Blowup}
Let $n \ge 3.$ Assume that $\Gamma(t)= ( \gamma^{(k)}_t )_{k \geq 1}$ solves the focusing (i.e. $\mu=-1$) Gross-Pitaevskii hierarchy \eqref{eq:GPHierarchyEquaFunct} with initial condition $\Gamma(0) = ( \gamma^{(k)}_0 )_{k \geq 1} \in {\mathfrak{H}^{1}}.$ If $( \gamma^{(k)}_t )_{k \geq 1}$ is a sequence of density operators such that $V_k (\Gamma(0)) < \8$ and $E_k (\Gamma(0))<0$ for some $k \ge 1,$ then the solution $\Gamma(t)$ blows up in finite time with respect to $\mathfrak{H}^1.$
\end{theorem}

\begin{remark}\label{rk:GPBlowup}\rm
We note that in the factorized case, Theorem \ref{th:Blowup} coincides with the corresponding result for the focusing GP equation \eqref{eq:GPEqua} (e.g., see Theorem 6.5.4 in \cite{C2003}).
\end{remark}

\begin{proof}
By \eqref{eq:VirialIdentity}, we have
\be\begin{split}
\partial^2_t V_k (\Gamma (t)) & = 16 k \mathrm{T r} \big [ - \frac{1}{2} \triangle_{x_1} \gamma^{(k)} \big ] + 2 n k \mu \mathrm{T r} \big [ B^{(k)}_{1, +} \gamma^{(k+1)}\big ]\\
& = 16 E_k (\Gamma (t)) + \mu k (2 n -4) \mathrm{T r} \big [ B^{(k)}_{1, +} \gamma^{(k+1)}\big ]\\
& = 16 E_k (\Gamma(0)) - k(2 n  -4) \mathrm{T r} \big [ B^{(k)}_{1, +} \gamma^{(k+1)}\big ]\\
\end{split}\ee
where we have used Theorem \ref{th:k-energyconservation}. Since $( \gamma^{(k)}_t )_{k \geq 1}$ is a sequence of nonnegative operators and $n \ge 3,$ we have
\beq\label{eq:V}
\partial^2_t V_k (\Gamma (t)) \le 16 E_k (\Gamma(0)).
\eeq
Since $V_k (\Gamma (t))$ is nonnegative, we conclude from the assumption $E_k (\Gamma(0)) <0$ that there exists a finite time $T^*$ such that $V_k (\Gamma (t)) \searrow 0$ as $t \nearrow T^*.$

On the other hand, we have
\be
\begin{split}
1 = {\rm T r} [ \gamma^{(k)} (t) ] & \leq \Big ( {\rm T r} \big [ |x_1|^2 \gamma^{(k)} (t) \big ] \Big )^{\frac{1}{2}} \Big ( {\rm T r} \big [ \frac{1}{|x_1|^2} \gamma^{(k)}(t) \big ] \Big )^{\frac{1}{2}}\\
& \leq C \Big ( {\rm T r} \big [ |x_1|^2 \gamma^{(k)} (t) \big ] \Big )^{\frac{1}{2}} \Big( {\rm T r} \big [ (- \Delta_{x_1})^{\frac{1}{2}} (- \Delta_{x'_1})^{\frac{1}{2}} \gamma^{(k)} (t) \big ] \Big )^{\frac{1}{2}}\\
\end{split}
\ee
where we have first used the Cauchy-Schwarz inequality, and then the Hardy inequality (noticing that $n > 2$). Thus,
\be
\mathrm{T r} \big [ - \Delta_{x_1} \gamma^{(k)} (t) \big ] = \mathrm{T r} \big [ (- \Delta_{x_1})^{\frac{1}{2}} (- \Delta_{x'_1})^{\frac{1}{2}} \gamma^{(k)} (t) \big ] \gtrsim \frac{1}{V_k (\Gamma (t))} \to \8 \quad \text{as}\quad t \to T^*,
\ee
where we have used \eqref{eq:identity} in the first equality.

Now, by \eqref{normH} and \eqref{eq:SobolevSpaceNorm}, one has
\begin{equation*}
\begin{split}
\|\Gamma(t)\|^k_{\mathfrak{H}^1} \geq ||| \gamma^{(k)} (t) |||_k &  = {\rm T r} \big [ S^{(k)} \gamma^{(k)} (t) \big ] = \mathrm{T r} \Big [ \prod^k_{j =1} (1 - \Delta_{x_j}) \gamma^{(k)} (t) \Big ]\\
& \ge  \sum^k_{j=1} {\rm T r} \Big [ (- \Delta_{x_j}) \gamma^{(k)} (t) \Big ] = k {\rm T r} \Big [ -\Delta_{x_1} \gamma^{(k)} (t) \Big ] \to \infty
\end{split}
\end{equation*}
as $t\nearrow T^*,$ which establishes blowup in finite time.
\end{proof}

\section{The quintic Gross-Pitaevskii hierarchy}\label{GPQuintic}

In this section, we consider the so-called quintic Gross-Pitaevskii hierarchy and present the corresponding results similar to the ones obtained for the cubic case in previous sections.

\subsection{Preliminaries}

Recall that the quintic Gross-Pitaevskii (or, GP in short) hierarchy $\Gamma(t)= ( \gamma^{(k)}(t) )_{k\geq 1}$ is given by
\beq\label{eq:GPHierarchyEquaQuintic}
\mathrm{i} \partial_t \gamma^{(k)}_t = \big [\sum^k_{j=1} (- \Delta_{x_j}), \gamma^{(k)}_t \big ] + \mu Q^{(k)} \gamma^{(k+2)}_t,\quad \mu = \pm 1,
\eeq
in $n$ dimensions, for $k \in \mathbb{N},$ where the operator $Q^{(k)}$ is defined by
\be
Q^{(k)} \gamma^{(k+2)}_t = \sum^k_{j=1} \mathrm{T r}_{k+1, k+2} \left [ \delta (x_j - x_{k+1}) \delta (x_j - x_{k+2}), \gamma^{(k+2)}_t \right ].
\ee
It is {\it defocusing} if $\mu =1,$ and {\it focusing} if $\mu= -1.$ We note that the quintic GP hierarchy accounts for $3$-body interactions between the Bose particles (see \cite{CP2011} and references therein for details).

In terms of kernel functions, the Cauchy problem for the quintic GP hierarchy \eqref{eq:GPHierarchyEquaQuintic} can be written as follows
\beq\label{eq:GPHierarchyEquaFunctQuintic}
\left \{ \begin{split}
& \big ( \mathrm{i} \partial_t + \triangle^{(k)} \big ) \gamma^{(k)}_t ({\bf x}_k;{\bf x}'_k) = \mu \big ( Q^{(k)} \gamma^{(k+2)}_t \big ) ({\bf x}_k; {\bf x}'_k ),\\
& \gamma^{(k)}_{t=0} ({\bf x}_k;{\bf x}'_k) = \gamma^{(k)}_0 ({\bf x}_k;{\bf x}'_k),\; k \in \mathbb{N},
\end{split}\right.
\eeq
where $Q^{(k)}: = \sum^k_{j=1} Q^{(k)}_j$ with the action of $Q^{(k)}_j = Q^{(k)}_{j, +} - Q^{(k)}_{j, -}$ on $\gamma^{(k+2)} ({\bf x}_{k+2}, {\bf x}'_{k+2}) \in \mathcal{S} (\mathbb{R}^{(k+2)n} \times \mathbb{R}^{(k+2)n})$ being defined according to
\be\begin{split}
 \big ( Q^{(k)}_{j, +} & \gamma^{(k+2)} \big ) ({\bf x}_k, {\bf x}'_k)\\
& : = \int d x_{k+1} d x_{k+2} d x'_{k+1} d x'_{k + 2} \gamma^{(k+2)} ({\bf x}_k, x_{k+1}, x_{k+2}; {\bf x}'_k, x'_{k+1}, x'_{k+2}) \\
& \; \quad \times \prod^{k+2}_{\el=k+1} \delta (x_j - x_{\el}) \delta (x_j - x'_{\el})\\
& = \gamma^{(k+2)} ({\bf x}_k, x_j, x_j; {\bf x}'_k, x_j, x_j ),
\end{split}\ee
and
\be\begin{split}
 \big ( Q^{(k)}_{j, -} & \gamma^{(k+2)} \big ) ({\bf x}_k, {\bf x}'_k)\\
& : = \int d x_{k+1} d x_{k+2} d x'_{k+1} d x'_{k + 2} \gamma^{(k+2)} ({\bf x}_k, x_{k+1}, x_{k+2}; {\bf x}'_k, x'_{k+1}, x'_{k+2}) \\
& \; \quad \times \prod^{k+2}_{\el=k+1} \delta (x'_j - x_{\el}) \delta (x'_j - x'_{\el})\\
& = \gamma^{(k+2)} ({\bf x}_k, x'_j, x'_j; {\bf x}'_k, x'_j, x'_j ),
\end{split}\ee
for $j=1, \ldots, k.$

Let $\varphi \in \mathrm{H}^1(\mathbb{R}^n),$ then one can easily verify that a particular solution to \eqref{eq:GPHierarchyEquaFunctQuintic} with initial conditions
\be
\gamma^{(k)}_{t=0}({\bf x}_k; {\bf x}'_k) = \prod^k_{j=1} \varphi(x_j) \overline{\varphi(x'_j)},\quad k=1,2,\ldots,
\ee
is given by
\be
\gamma^{(k)}_t ({\bf x}_{k}; {\bf x}'_{k} ) = \prod^k_{j=1} \varphi_t (x_j) \overline{\varphi_t ( x'_j )},\quad k=1,2,\ldots,
\ee
where $\varphi_t$ satisfies the quintic non-linear Schr\"odinger equation
\beq\label{eq:GPEquaQuintic}
\mathrm{i} \partial_t \varphi_t = -\Delta \varphi_t + \mu |\varphi_t|^4 \varphi_t,\quad \varphi_{t=0}=\varphi.
\eeq

The GP hierarchy \eqref{eq:GPHierarchyEquaFunctQuintic} can be written in the integral form
\beq\label{eq:GPHierarchyIntEquaQuintic}
\gamma^{(k)}_t = e^{ \mathrm{i} t{\Delta}^{(k)}}  \gamma^{(k)}_0 + \int^{t}_{0} d s~ e^{ \mathrm{i} (t-s){\Delta}^{(k)}}  \tilde{Q}^{(k)} \gamma^{(k+2)}_s,\; k=1,2,\ldots,
\eeq
where $\tilde{Q}^{(k)} = - \mathrm{i} \mu Q^{(k)}.$ Evidently, such a solution can be obtained by solving the following infinity linear hierarchy of integral equations
\begin{equation}\label{eq:GPStrongSolutionFunctQuintic}
\tilde Q^{(k)}\gamma^{(k+2)}_t = \tilde Q^{(k)}e^{ \mathrm{i} t{\Delta}^{(k+2)}} \gamma^{(k+2)}_0 + \int^{t}_{0} d s~ \tilde{Q}^{(k)}e^{ \mathrm{i} (t-s){\Delta}^{(k+2)}} \tilde{Q}^{(k+2)}\gamma^{(k+4)}_s,
\end{equation}
for any $k\geq 1.$ If we write
\be
\hat \Delta \Gamma:= ( \Delta^{(k)}\gamma^{(k)} )_{k\geq 1} \quad \text{and} \quad \hat{Q} \Gamma: = ( \tilde{Q}^{(k)}\gamma^{(k+2)} )_{k\geq 1} ,
\ee
then \eqref{eq:GPHierarchyIntEquaQuintic} and \eqref{eq:GPStrongSolutionFunctQuintic} can be written as
\begin{equation}\label{Gamma:1}
\Gamma(t)=e^{ \mathrm{i} t\hat{\Delta}}\Gamma_0+\int_0^tds~ e^{ \mathrm{i} (t-s)\hat{\Delta}} \hat{Q} \Gamma(s),
\end{equation}
and
\begin{equation}\label{Gamma:2}
\hat Q\Gamma(t) =\hat Qe^{ \mathrm{i} t\hat{\Delta}}\Gamma_0+\int_0^tds~ \hat Q e^{ \mathrm{i} (t-s)\hat{\Delta}}\hat Q\Gamma(s),
\end{equation}
respectively.

Formally we can expand the solution $\gamma^{(k)}_t$ of \eqref{eq:GPHierarchyIntEquaQuintic}
for any $m \geq 2$ as
\beq\label{eq:DuhamelExpanQuintic}
\begin{split}
\gamma^{(k)}_t = & e^{ \mathrm{i} t{\Delta}^{(k)}}  \gamma^{(k)}_0 + \sum^{m-1}_{j=1} \int^t_0 d t_2 \int^{t_2}_0 d t_4 \cdots
\int^{t_{2(j-1)}}_0 d t_{2 j}  e^{ \mathrm{i} (t-t_2 ){\Delta}^{(k)}} \tilde{Q}^{(k)} \cdots\\
& \; \times e^{ \mathrm{i} ( t_{2 (j-1)} - t_{2 j} ){\Delta}^{(k+2(j-1))}} \tilde{Q}^{(k+2(j-1))}   e^{ \mathrm{i} t_j {\Delta}^{(k+2j)}}\gamma^{(k+2j)}_0\\
& \; + \int^t_0 d t_2 \int^{t_2}_0 d t_4 \cdots \int^{t_{2(m-1)}}_0 d t_{2 m} e^{ \mathrm{i} (t-t_2){\Delta}^{(k)}} \tilde{Q}^{(k)} \cdots \\
& \; \times e^{ \mathrm{i} ( t_{2(m-1)} - t_{2 m} ){\Delta}^{(k+2(m-1))}} \tilde{Q}^{(k+2(m-1))} \gamma^{(k+2m)}_{t_m},
\end{split}
\eeq
with the convention $t_0 =t.$

\subsection{Local well-posedness}

In this subsection, we present the local well-posedness results for the quintic GP hierarchy \eqref{eq:GPHierarchyEquaFunctQuintic}.

Let us make the notion of solution more precise.

\begin{definition}\label{df:StrongSolutionQuintic}
A function $\Gamma (t) = ( \gamma^{(k)}_t )_{k \geq 1}: I \mapsto \mathcal{H}^{s}$ on a non-empty time interval $0 \in I \subset \mathbb{R}$ is said to be a local $(strong)$ solution
to the Gross-Pitaevskii hierarchy \eqref{eq:GPHierarchyEquaFunctQuintic} if it lies in the class $C (K, \mathcal{H}^{s})$ for all compact sets $K \subset I$ and obeys the Duhamel formula
\beq\label{eq:MildSolution}
\gamma^{(k)}_t = e^{\mathrm{i} t \triangle^{(k)}} \gamma^{(k)}_0 - \mathrm{i} \mu \int^{t}_{0} d s\; e^{\mathrm{i} (t-s) \triangle^{(k)}} Q^{(k)} \gamma^{(k+2)}_s,\quad \forall t \in I,
\eeq
holds in $\mathrm{H}^{s}_k$ for every $k=1,2,\ldots.$
\end{definition}

The following theorem was stated in \cite{Chen}.

\begin{theorem}\label{th:LocalWellposuedQuintic1} {\rm (cf. \cite{Chen}, Theorem 5.1)}
Assume that $n \geq 1$ and $s > \frac{n}{2}.$  The Cauchy problem  \eqref{eq:GPHierarchyEquaFunctQuintic} is locally well posed. More precisely, there exists a constant $ K_{ n,s} > 0$ depending only on $n$ and $s$ such that
\begin{enumerate}[{\rm (1)}]

\item For every $\Gamma_0= ( \gamma_0^{(k)} )_{k\geq 1}\in \mathcal{H}^{s},$ we let $T= \frac{K_{n,s}}{\|\Gamma_0\|_{\mathcal{H}^{s}}^2}$ and $I = [-T, T].$ Then there exists a solution $\Gamma(t)= ( \gamma_t^{(k)} )_{k\geq 1}\in C(I,\mathcal{H}^{s})$ to the Gross-Pitaevskii hierarchy \eqref{eq:GPHierarchyEquaFunctQuintic} with the initial data $\Gamma_0$ satisfying
\begin{equation}
\|\Gamma(t)\|_{C(I,\mathcal{H}^{s})}\leq 2 \|\Gamma_0\|_{\mathcal{H}^{s}}
\end{equation}

\item Given $I_0 = [-T_0, T_0 ]$ with $T_0>0.$ If $\Gamma(t),\Gamma'(t)$  in $C(I_0, \mathcal{H}^{s})$ are two solutions to the Gross-Pitaevskii
hierarchy \eqref{eq:GPHierarchyEquaFunctQuintic} with the initial conditions $\Gamma (0) =\Gamma_0$ and  $\Gamma' ( 0 ) = \Gamma_0'$ in $\mathcal{H}^{s},$ respectively, then
\begin{equation}
\|\Gamma(t)-\Gamma'(t)\|_{C(I, \mathcal{H}^{s})} \leq 2 \|\Gamma_0-\Gamma_0'\|_{\mathcal{H}^{s}}
\end{equation}
with $I = [-T, T],$ where
\be
T = \min \left \{ T_0,\; \frac{K_{n,s}}{\|\Gamma(t)-\Gamma'(t)\|_{C(I_0, \mathcal{H}^{s})}^2} \right \}.
\ee
\end{enumerate}
\end{theorem}

The proof can be obtained by use of the fully expanded iterated Duhamel series and a Cauchy convergence criterion, based on the following inequality
\be
\| Q^{(k)} \gamma^{(k+2)} \|_{\mathrm{H}^{s}_k} \le C_{n, s} k \| \gamma^{(k+2)} \|_{\mathrm{H}^{s}_{k+2}},\quad \forall k \ge 1,
\ee
with $C_{n, s}>0$ being a constant depending only on $n$ and $s,$ which was proved in \cite{CP2011} (Theorem 4.3 there).

For the case $s \le n/2$ we have

\begin{theorem}\label{th:LocalWellposuedQuintic2}
Assume that $n \ge 2$ and $s > \frac{n-1}{2}.$ Then, the Cauchy problem for the Gross-Pitaevskii hierarchy \eqref{eq:GPHierarchyEquaFunctQuintic}  is locally well posed in $\mathcal{H}^{s}.$ More precisely, there exist an absolute constant $A>2$ and a constant $C=M_{n, s}>0$ depending only on $n$ and $s$ such that
\begin{enumerate}[{\rm (1)}]

\item For every $\Gamma_0 = ( \gamma^{(k)}_{0} )_{k \geq 1} \in \mathcal{H}^{s},$ we let $T = \frac{M_{n, s}}{\| \Gamma_0 \|^4_{\mathcal{H}^{s}}}$ and $I = [- T, T].$ Then there exists a solution $\Gamma (t) = ( \gamma^{(k)} (t) )_{k \geq 1} \in C ( I, \mathcal{H}^{s})$ to \eqref{eq:GPHierarchyEquaFunctQuintic} with the initial data $\Gamma (0) = \Gamma_0$ such that
\begin{equation}\label{eq:SpacetimeEstimateQuintic}
\| \hat{Q} \Gamma (t) \|_{L^1_{t \in I} \mathcal{H}^{s}} \leq 2 A \| \Gamma_0 \|_{\mathcal{H}^{s}}.
\end{equation}

\item Given $I_0 = [- T_0, T_0]$ with $T_0 >0.$ If $\Gamma (t) \in C (I_0, \mathcal{H}^{s})$ so that $\hat{Q} \Gamma (t) \in L^1_{t \in I_0} \mathcal{H}^{s}$ is a solution to \eqref{eq:GPHierarchyEquaFunctQuintic} with the initial data $\Gamma (0) = \Gamma_0,$ then \eqref{eq:SpacetimeEstimateQuintic} holds true as well for $I = [-T, T],$ where
\be
T = \min \left \{ T_0, \; \frac{M_{n, s}}{ \| \hat{Q} \Gamma (t) \|^4_{L^1_{t \in I_0} \mathcal{H}^{s}} + \| \Gamma_0 \|^4_{\mathcal{H}^{s}} } \right \}.
\ee

\item Given $I_0 = [-T_0, T_0]$ with $T_0 >0.$ If $\Gamma (t)$ and $\Gamma' (t)$ in $C ( I_0, \mathcal{H}^{s})$ with $\hat{Q} \Gamma (t), \hat{Q} \Gamma' (t) \in L^1_{t \in I_0}\mathcal{H}^{s}$ are two solutions to \eqref{eq:GPHierarchyEquaFunctQuintic} with initial conditions $\Gamma (0) = \Gamma_0$ and $\Gamma' (0) = \Gamma'_0$ in $\mathcal{H}^{s},$ respectively, then
\begin{equation}\label{eq:InitialValueContinuousDependence}
\| \Gamma (t) - \Gamma' (t) \|_{C( I, \mathcal{H}^{s})} \leq (1+ 2 A) \| \Gamma_0 - \Gamma'_0\|_{\mathcal{H}^{s}},
\end{equation}
with $I = [-T, T],$ where
\be
T = \min \left \{ T_0,\; \frac{M_{n, s}}{ \| \hat{Q} [\Gamma (t) - \Gamma' (t) ]\|^4_{L^1_{t \in I_0} \mathcal{H}^{s}} + \| \Gamma_0 - \Gamma'_0 \|^4_{\mathcal{H}^{s}} } \right \}.
\ee
\end{enumerate}
\end{theorem}

This result was announced in \cite{Chen}. For the sake of completeness, we present the proof. To this end, we need two preliminary results

\begin{lemma}\label{le:QOperatorEstimate}
Assume that $n\geq 2$ and $s> \frac{n-1}{2}.$ Then there exists a constant $C_{n,s}>0$ depending only on $n$ and $s$ such that, for any symmetric $\gamma^{(k+2)}\in {\mathcal S}({\mathbb R}^{(k+2)n}\times{\mathbb R}^{(k+2)n} ),$
\begin{equation}
\|Q^{(k)}_{j} e^{\mathrm{i} t\Delta^{(k+2)}}\gamma^{(k+2)}\|_{L_t^2({\mathrm H}_k^{s})}\leq C_{n,s}\|\gamma^{(k+2)}\|_{{\mathrm H}_{k+2}^{s}}
\end{equation}
for all $k\geq 1,$ where $j=1,2,\ldots,k.$

Consequently, $Q^{(k)}$ can be extended to the space ${\mathrm H}_{k+2}^{s}$ such that
\begin{equation}
\|Q^{(k)}e^{\mathrm{i} t\Delta^{(k+1)}}\gamma^{(k+2)}\|_{L_t^2({\mathrm H}_k^{s})}\leq C_{n,s}k\|\gamma^{(k+2)}\|_{{\mathrm H}_{k+2}^{s}}
\end{equation}
for all $\gamma^{(k+2)}\in {\mathrm H}_{k+2}^{s}.$
\end{lemma}

This lemma was proved first for $n=3$ in \cite{KM2008}, and then in \cite{CP2010} (Proposition A.1 there) for general case.

For any $\Gamma = (\gamma^{(k)}_t )_{k \ge 1}$ we define
\be\begin{split}
P_{k+2, j} (\Gamma) (t) : = & \int^t_0 d t_2 \int^{t_2}_0 d t_4 \cdots \int^{t_{2(j-1)}}_0 d t_{2 j} e^{\mathrm{i} (t - t_2) \triangle^{(k+2)}} \tilde{Q}^{(k+2)} \cdots \\
& \times e^{\mathrm{i} (t_{2(j-1)} - t_{2 j}) \triangle^{(k+ 2 j)}} \tilde{Q}^{(k+ 2 j)} e^{\mathrm{i} t_{2 j} \triangle^{(k+2(j+1))}} \gamma^{(k+2(j+1))}_{t_{2 j}}
\end{split}\ee
with the convention $t = t_0.$

The following lemma is crucial for the proof of Theorem \ref{th:LocalWellposuedQuintic2}.

\begin{lemma}\label{le:DuhamelEstimate}
Assume that $n\geq 2$ and $s> \frac{n-1}{2}.$  Then there exists an absolute constant $A>2$ and a constant $C_{n,s}$ depending only on $n$ and $s$ so that the estimates below hold
\begin{enumerate}[{\rm (1)}]

\item  For any $\Gamma_0 = ( \gamma_0^{(k)} )_{k\geq 1}\in \bigotimes_{k=1}^{\infty} {\mathrm H}_{k}^{s},$
\beq\label{eq:QoperatorEstimation1}
\|\tilde{Q}^{(k)} P_{k+2, j} ( e^{\mathrm{i} t \Delta} \Gamma_0) (t)\|_{L_{t\in I}^1{\mathrm H}_{k}^{s}}\leq kA^{k+j}(C_{n,s}T)^{\frac{j+1}{2}}\|\gamma^{(k+2j+2)}_0\|_{{\mathrm H}_{k+2j+2}^{s}}
\eeq
for $k,j\geq 1$ and $T>0,$ where $I =[-T, T]$ and $ e^{\mathrm{i} t \Delta} \Gamma_0 = ( e^{\mathrm{i} t\Delta^{(k)}}\gamma_0^{(k)} )_{k\geq 1}.$

\item  For any $T>0$ and $\Gamma(t)= ( \gamma^{(k)}_t )_{k \geq 1}$  with $\gamma^{(k)}_t \in L_{t \in [-T,T]}^1{\mathrm H}_{k}^{s},$
\beq\label{eq:QoperatorEstimation2}
\begin{split}
\|\tilde{Q}^{(k)}& P_{k+2, m} (\Gamma )(t)\|_{L_{t \in I}^1{\mathrm H}_{k}^{s}}\\ &\leq kA^{k+m}(C_{n,s}T)^{\frac{m}{2}}\|Q^{(k+2m)}\gamma^{(k+2m+2)}(t)\|_{L_{t \in I}^1{\mathrm H}_{k+2m}^{s}},
\end{split}
\eeq
for $k,m\geq 1,$ where $I =[-T, T].$

\end{enumerate}
\end{lemma}

\begin{proof}
The inequalities \eqref{eq:QoperatorEstimation1} and \eqref{eq:QoperatorEstimation2} can be proved by using the so-called "board game" argument presented in \cite{KM2008}. For the details see the proof of Theorem 6.2 in \cite{CP2011} (e.g., Proposition A.2 in \cite{CP2010}).
\end{proof}

Now we are ready to prove Theorem \ref{th:LocalWellposuedQuintic2}. To this end, we introduce the system
\begin{equation}
\Gamma(t)=e^{\mathrm{i} t\hat{\Delta}}\Gamma_0+\int_0^tds~ e^{\mathrm{i} (t-s)\hat{\Delta}}\Xi_s
\end{equation}

\begin{equation}\label{eq:GPStrongSolutionFunctXi}
\Xi_t=\hat Qe^{\mathrm{i} t\hat{\Delta}}\Gamma_0+\int_0^tds~ \hat Q e^{\mathrm{i} (t-s) \hat{\Delta}} \Xi_s
\end{equation}
which is formally equivalent to the system  \eqref{Gamma:1} and \eqref{Gamma:2}.

The proof is divided into three parts as follows.

\begin{proof}
(1).\; Let  $n\geq 2$ and $s>(n-1)/2.$ Let $\Gamma_0=\{\gamma_0^{(k)}\}_{k\geq 1}\in \mathcal{H}^{s}$ and
$\Xi_0=\{\rho_0^{(k)}(t)\}_{k\geq 1}=0.$ Given $k\geq 1,$ for any $m\geq 1$ we define
\begin{equation}\label{eq:m-iteraPf}
\rho_m^{(k)}(t)= \tilde Q^{(k)}e^{\mathrm{i} t{\Delta}^{(k+2)}} \gamma^{(k+2)}_0 + \int^{t}_{0} d s~ \tilde{Q}^{(k)}e^{\mathrm{i} (t-s){\Delta}^{(k+2)}}\rho_{m-1}^{(k+2)}(s),
\end{equation}
for $t \in I =[-T , T],$ where $T > 0$ will be fixed later. Set $\Xi_m(t)=\{\rho_m^{(k)}(t)\}_{k\geq 1}$ for every
$m\geq 1.$ By iterating \eqref{eq:m-iteraPf}, for every $m\geq 2$ one has
\be
\begin{split}
\rho_m^{(k)} (t) = \sum^{m-1}_{j=0}  \tilde Q^{(k)} P_{k+2, j} (e^{\mathrm{i} t \Delta} \Gamma_0) (t)
\end{split}
\ee
with the convenience $P_{k+2, 0} ( e^{\mathrm{i} t \Delta} \Gamma_0)(t) = e^{\mathrm{i} t\Delta^{(k+2)}}\gamma_0^{(k+2)}.$ By Lemma \ref{le:DuhamelEstimate} (1) we have
\be
\begin{split}
\Big \|\sum_{j=0}^{m-1} & \tilde{Q}^{(k)} P_{k+2, j} (e^{\mathrm{i} t \Delta} \Gamma_0) (t) \Big \|_{L_{t \in I}^1{\mathrm H}_{k}^{s}}\\
&\leq k
\sum_{j=0}^{m-1}A^{k+j}(\sqrt{C_{n,s}T})^{j+1}\|\gamma_0^{(k+2j+2)}\|_{{\mathrm H}_{k+2j+2}^{s}}
\end{split}
\ee
i.e.,
\begin{equation}\label{pho}
\|\rho_m^{(k)}\|_{L_{t\in t}^1{\mathrm H}_{k}^{s}}\leq k
\sum_{j=0}^{m-1}A^{k+j}(\sqrt{C_{n,s}T})^{j+1}\|\gamma_0^{(k+2j+2)}\|_{{\mathrm H}_{k+2j+2}^{s}}.
\end{equation}
Set $T : = \frac{M_{n, s}}{\|\Gamma_0\|_{\mathcal{H}^{s}}^4}$ with $M_{n, s} = \frac{1}{8C_{n,s}A^2}.$  For $\lambda>0$ one has by \eqref{pho}
\be\begin{split}
\sum_{k\geq 1}&\frac{1}{\lambda^k}\|\rho_m^{(k)}\|_{L_{t\in I}^1{\mathrm H}_{k}^{s}}\\
&\leq \frac{1}{A} \sum_{k\geq 1}k\Big(\frac{ A}{\lambda}\Big)^k\sum_{j=0}^{m-1}(A\sqrt{C_{n,s}T})^{j+1}\|\gamma_0^{(k+2j+2)}\|_{{\mathrm H}_{k+2j+2}^{s}}\\
&\leq\frac{1}{A} \sum_{k\geq 1}k\Big(\frac{ A}{\lambda}\Big)^k\sum_{\el=k+2}^{\infty}(A\sqrt{C_{n,s}T})^{\frac{\el - k}{2}}\|\gamma_0^{(\el)}\|_{{\mathrm H}_{\el}^{s}}\\
&= \frac{1}{A}\sum_{k\geq 1}k\Big(\frac{ A}{\lambda}\Big)^k\sum_{\el=k+2}^{\infty}\frac{1}{(\|\Gamma_0\|_{\mathcal{H}^{s}})^{\el-k}}\|\gamma_0^{(\el)}\|_{{\mathrm H}_{\el}^{s}}\\
&\leq\frac{1}{A} \sum_{k\geq 1}k\Big(\frac{A}{\lambda} \Big)^ k (\|\Gamma_0\|_{\mathcal{H}^{s}})^{k} \sum_{\el=1}^{\infty} \frac{1}{(2\|\Gamma_0 \|_{\mathcal{H}^{s}})^{\el}} \|\gamma_0^{(\el)}\|_{{\mathrm H}_{\el}^{s}}\\
&\leq\frac{1}{A} \sum_{k\geq 1}k \Big(\frac{A\|\Gamma_0\|_{\mathcal{H}^{s}} }{\lambda}\Big)^k =\frac{1}{A}\frac{\xi}{(1-\xi)^2},
\end{split}
\ee
with $\xi= A\|\Gamma_0\|_{\mathcal{H}^{s}}/\lambda.$  Note that $A>2,$ choosing $\xi=\frac{1}{2},$ in other words, $\lambda=2A\|\Gamma_0\|_{\mathcal{H}^{s}},$ we have
\be
\sum_{k\geq 1}\frac{1}{\lambda^k}\|\rho_m^{(k)}\|_{L_{t\in I}^1{\mathrm H}_{k}^{s}}
\leq\frac{1}{A}\frac{\xi}{(1-\xi)^2}\leq 1.
\ee
This concludes that for every $m \ge 1,$ $\Xi_m \in L^1_{t \in I} \mathcal{H}^{s}$ and
\begin{equation}\label{eq:m-iteraApriorBound}
\| \Xi_m \|_{L^1_{t \in I} \mathcal{H}^{s}} \le 2 A \| \Gamma_0 \|_{\mathcal{H}^{s}}.
\end{equation}

Now, for fixed $k \ge 1$ and any $n, m$ with $n > m$ we have
\be
\begin{split}
\| \rho^{(k)}_m & - \rho^{(k)}_n \|_{L^1_{t \in I} \mathrm{H}^{s}_k}\\
\le & \frac{k}{A}\sum^{n-1}_{j=m} A^{k+j} ( \sqrt{C_{n, s} T})^{j+1} \big \| \gamma^{(k+2j+2)}_0 \big \|_{\mathrm{H}^{s}_{k+2j+2}}\\
\le & \frac{k}{A}( A \| \Gamma_0 \|_{\mathcal{H}^{s}})^k \sum_{\el \ge k+2m+4} \frac{1}{( \| \Gamma_0 \|_{\mathcal{H}^{s}})^{\el}} \big \| \gamma^{(\el)}_0 \big \|_{\mathrm{H}^{s}_{l}}.
\end{split}
\ee
This concludes that for each $k \ge 1,$ $\rho^{(k)}_m$ converges in $L^1_{t \in I} \mathrm{H}^{s}_k$ as $m \to \8,$
whose limitation is denoted by $\rho^{(k)}.$

Set $\Xi (t) = \{\rho^{(k)} (t) \}_{k \ge 1}.$ Note that for any $m \ge 1,$
\be
\begin{split}
\Big \| \int^t_0 d s & \tilde{Q}^{(k)}  e^{\mathrm{i} (t-s) \triangle^{(k+2)}} [ \rho^{(k+2)}_{m-1} (s) - \rho^{(k+2)} (s) ]\Big \|_{L^1_{t \in I} \mathrm{H}^{s}_k}\\
\le & \sum^k_{\el =1} \int^T_{-T} \int^T_{-T} d t d s \big \| Q_{\el, k} e^{\mathrm{i} (t-s) \triangle^{(k+2)}} [ \rho^{(k+2)}_{m-1} (s) - \rho^{(k+2)} (s) ] \big \|_{\mathrm{H}^{s}_k}\\
\le & T^{1/2} \sum^k_{\el =1} \int^T_{-T} d s \big \| Q_{\el, k} e^{\mathrm{i} (t-s) \triangle^{(k+2)}} [ \rho^{(k+2)}_{m-1} (s) - \rho^{(k+2)} (s) ] \big \|_{L^2_{t \in I}\mathrm{H}^{s}_k}\\
\le & C_{n, s} k T^{1/2} \big \| \rho^{(k+2)}_{m-1} - \rho^{(k+2)} \big \|_{L^1_{t \in I} \mathrm{H}^{s}_{k+1}},
\end{split}
\ee
where we have used the Cauchy-Schwarz inequality with respect to the integral in $t$ in the second inequality and used Lemma \ref{le:QOperatorEstimate} in the last inequality. Thus, taking $m \to \8$ in \eqref{eq:m-iteraPf} we prove that $\Xi$ is a solution to \eqref{eq:GPStrongSolutionFunctXi}.
Moreover, taking $m \to \8$ in \eqref{eq:m-iteraApriorBound} we obtain \eqref{eq:SpacetimeEstimateQuintic}.

(2).\; Fix $T_0 >0.$ Suppose $\Gamma (t) \in C(I_0, \mathcal{H}^{s} )$ is a solution to \eqref{eq:GPHierarchyEquaFunctQuintic} so that $\hat{Q} \Gamma (t) \in L^1_{t \in [0,T_0]}\mathcal{H}^{s}.$ Let $ T \in (0, T_0],$ which will be fixed later. By iterating \eqref{eq:GPStrongSolutionFunctQuintic}, for every $m \geq 2$ one has
\be
\begin{split}
\tilde{Q}^{(k)}\gamma^{(k+2)}_t = & \tilde{Q}^{(k)} e^{\mathrm{i} t \triangle^{(k+2)}} \gamma^{(k+2)}_0\\
& \quad + \sum^{m-1}_{j=1} \tilde{Q}^{(k)} P_{k+2, j} (e^{\mathrm{i} t \Delta} \Gamma_0) (t) + \tilde{Q}^{(k)} P_{k+2, m} ( \Gamma) (t).
\end{split}
\ee
By Lemma \ref{le:DuhamelEstimate} we have
\begin{equation*}
\begin{split}
\big \| \tilde{Q}^{(k)} \gamma^{(k+2)}_t \big \|_{L^1_{t \in I} \mathrm{H}^{s}_k}
& \le \sum^{m-1}_{j=0} k A^{k+j} ( \sqrt{C_{n, s} T})^{j+1} \big \| \gamma^{(k+2j+2)}_0 \big \|_{\mathrm{H}^{s}_{k+2j+2}}\\
&\quad + k A^{k+m} (\sqrt{C_{n, s} T})^m \big \| Q^{(k+2m)} \gamma^{(k+2m+2)}_t \big \|_{L^1_{t \in I}\mathrm{H}^{s}_{k+2m}}.
\end{split}
\end{equation*}
Set
\be
T = \min \left \{ T_0, \; \frac{M_{n, s}}{ \| \hat{Q} \Gamma (t) \|^4_{L^1_{t \in [0,T_0]}\mathcal{H}^{s}} + \| \Gamma_0 \|^4_{\mathcal{H}^{s}} } \right \}.
\ee
Taking $m \to \8$ we have
\begin{equation}
\big \| \tilde{Q}^{(k)} \gamma^{(k+2)}_t \big \|_{L^1_{t \in I} \mathrm{H}^{s}_k} \le \frac{k}{A}(A \| \Gamma_0 \|_{\mathcal{H}^{s}})^k \sum^{\8}_{\el=1} \frac{1}{( \| \Gamma_0 \|_{\mathcal{H}^{s}})^{\el}} \big \| \gamma^{(\el)}_0 \big \|_{\mathrm{H}^{s}_{\el}}
\end{equation}
Then, for $\lambda>0$ we have
\be
\begin{split}
\sum_{k \ge 1} \frac{1}{\lambda^k} & \big \|\tilde{Q}^{(k)} \gamma^{(k+2)}_t \big \|_{L^1_{t \in I} \mathrm{H}^{s}_k} \\
\le & \frac{1}{A}\sum_{k \ge 1}  k\Big ( \frac{A\| \Gamma_0 \|_{\mathcal{H}^{s}}}{\lambda}\Big )^k\sum^{\8}_{\el=1} \frac{1}{( \| \Gamma_0 \|_{\mathcal{H}^{s}})^{\el}} \big \| \gamma^{(\el)}_0 \big \|_{\mathrm{H}^{s}_{\el}} \leq \frac{\xi}{A(1-\xi)^2},
\end{split}
\ee
with $\xi= A\|\Gamma_0\|_{\mathcal{H}^{s}} /\lambda.$  Note that $A>2,$ choosing $\lambda = 2 A \|\Gamma_0\|_{\mathcal{H}^{s}},$ that is, $\xi=\frac{1}{2},$ we have
\be
\sum_{k \ge 1} \frac{1}{\lambda^k}  \big \| \tilde{Q}^{(k)} \gamma^{(k+2)}_t \big \|_{L^1_{t \in I} \mathrm{H}^{s}_k}\leq
 \frac{1}{A}\frac{\xi}{(1-\xi)^2}\leq 1.
\ee
This completes the proof of (2).

(3).\; Fix $T_0 >0$ with $I_0 = [- T_0, T_0].$ Suppose $\Gamma (t), \Gamma'(t) \in C (I_0, \mathcal{H}^{s})$ are two solutions to \eqref{eq:GPHierarchyEquaFunctQuintic} such that $\hat{Q} \Gamma (t), \hat{Q} \Gamma' (t) \in L^1_{t \in I_0}\mathcal{H}^{s}$ with $\Gamma (0) = \Gamma_0$ and $\Gamma' (0) = \Gamma'_0$ in $\mathcal{H}^{s},$ respectively. Since \eqref{eq:GPHierarchyEquaFunctQuintic} is linear, it suffices to consider $\Gamma (t)$ instead of $\Gamma (t) - \Gamma'(t).$ Set
\be
T = \min \left \{ T_0, \; \frac{M_{n, s}}{ \| \hat{Q} \Gamma (t) \|^4_{L^1_{t \in I_0}\mathcal{H}^{s}} + \| \Gamma_0 \|^4_{\mathcal{H}^{s}} } \right \}.
\ee
Then, by (2) we have for $I = [-T, T],$
\be
\| \Gamma (t) \|_{C (I, \mathcal{H}^{s})} \le \|\Gamma_0 \|_{\mathcal{H}^{s}} + \| \hat{Q} \Gamma (t) \|_{L^1_{t \in I}\mathcal{H}^{s}} \le (1+ 2 A) \|\Gamma_0 \|_{\mathcal{H}^{s}}.
\ee
This completes the proof.
\end{proof}

\subsection{Blowup alternative and blowup rate in finite time}

By slightly repeating the proofs of Theorems \ref{th:BlowupRate-alpha>n/2} and \ref{th:BlowupRate-alpha>(n-1)/2} with the help of Theorems \ref{th:LocalWellposuedQuintic1} and \ref{th:LocalWellposuedQuintic2}, respectively, we can obtain two results concerning the blowup alternative of solutions to the quintic GP hierarchy \eqref{eq:GPHierarchyEquaFunctQuintic} as follows.

\begin{theorem}\label{th:CauchyProbQuintic}
Assume that $n\geq 1$ and $s > \frac{n}{2}.$ If $\Gamma(t)$ is a solution of the Gross-Pitaevskii hierarchy \eqref{eq:GPHierarchyEquaQuintic} with initial condition $\Gamma(0)\in \mathcal{H}^{s}$ such that $T_{\mathrm{max}}< \8$ (resp. $T_{\mathrm{min}}< \8$), then $\lim_{t \nearrow T_{\mathrm{max}}} \| \Gamma (t) \|_{\mathcal{H}^{s}} = \8$ (resp. $\lim_{t \searrow - T_{\mathrm{min}}} \| \Gamma (t) \|_{\mathcal{H}^{s}} = \8$), and the following lower bound on the blowup rate holds
\beq\label{eq:BelowuprateQ1}\begin{split}
\|\Gamma(t)\|_{\mathcal{H}^{s}} & \geq \frac{K_{n,s}}{ | T_{\mathrm{max}}-t |^{\frac{1}{2}}},\quad \forall 0< t < T_{\mathrm{max}},\\
\Big ( \|\Gamma(- t)\|_{\mathcal{H}^{s}} & \geq \frac{K_{n,s}}{ | T_{\mathrm{min}}-t |^{\frac{1}{2}}},\quad \forall 0< t < T_{\mathrm{min}}\Big ),
\end{split}\eeq
where $K_{n,s}>0$ is a constant depending only on $n$ and $s.$
\end{theorem}

\begin{theorem}\label{th:CauchyProbQuintic2}
Assume that $n \geq 2$ and $s > \frac{n-1}{2}.$ If $\Gamma(t)$ is a solution of the Gross-Pitaevskii hierarchy \eqref{eq:GPHierarchyEquaQuintic} with initial condition $\Gamma(0)\in \mathcal{H}^{s}$ such that $T_{\mathrm{max}}< \8$ (resp. $T_{\mathrm{min}}< \8$), then $\lim_{t \nearrow T_{\mathrm{max}}} \| \Gamma (t) \|_{\mathcal{H}^{s}} = \8$ (resp. $\lim_{t \searrow - T_{\mathrm{min}}} \| \Gamma (t) \|_{\mathcal{H}^{s}} = \8$), and the following lower bound on the blowup rate holds
\beq\label{eq:BelowuprateQ2}\begin{split}
\|\Gamma(t)\|_{\mathcal{H}^{s}} & \geq \frac{M_{n, s}}{|T_{\mathrm{max}}-t|^{\frac{1}{4}}},\quad \forall 0< t < T_{\mathrm{max}},\\
\Big ( \|\Gamma( - t)\|_{\mathcal{H}^{s}} & \geq \frac{M_{n, s}}{|T_{\mathrm{min}}-t|^{\frac{1}{4}}},\quad \forall 0< t < T_{\mathrm{min}} \Big ),
\end{split}\eeq
where $M_{n,s}>0$ is a constant depending only on $n$ and $s.$
\end{theorem}

We omit the details of the proofs.

\subsection{Blowup in finite time} In this subsection, we present energy conservation and Virial type identities of solutions to \eqref{eq:GPHierarchyEquaFunctQuintic}. As an application, we can obtain a result on blowup in finite time of solutions to the focusing GP hierarchy \eqref{eq:GPHierarchyEquaFunctQuintic}. The arguments are the same as those involved in Sections \ref{EnergyVirialIdentity} and \ref{BlowupFinite}. The required modifications are not difficult and left to the interested reader. However, for the sake of convenience, we write the corresponding definitions and results.

First of all, as in \cite{CPT2010} we introduce the energy functional of solutions to \eqref{eq:GPHierarchyEquaFunctQuintic} as follows.
\be
E_k (\Gamma(t)): = \frac{1}{2} \mathrm{T r} \Big [ \sum^k_{j=1} (- \triangle_{x_j}) \gamma^{(k)}\Big ] + \frac{\mu}{6} \mathrm{T r} \Big [ \sum^k_{j=1} Q^{(k)}_{j, +} \gamma^{(k)}\Big ]
\ee
for any $k \ge 1.$ Then

\begin{theorem}\label{th:k-energyconservationQuintic}
Assume that $\Gamma(t)$ is a solution of the Gross-Pitaevskii hierarchy \eqref{eq:GPHierarchyEquaFunctQuintic} with initial condition $\Gamma(0)\in \mathcal{H}^{s}.$ Then $E_k (\Gamma (t))$ is a conserved quantity, i.e.,
\beq\label{eq:k-energyconversationQuintic}
E_k (\Gamma (t)) = E_k (\Gamma (0)), \quad \forall t \in (-T_{\mathrm{min}}, T_{\mathrm{max}}).
\eeq
\end{theorem}

To state the Virial type identity for solutions to \eqref{eq:GPHierarchyEquaFunctQuintic}, we set
\be
V_k (\Gamma (t)): = \mathrm{T r} \big [ | \mathbf{x}_k |^2 \gamma^{(k)} \big ] = \sum^k_{j =1} \int d \mathbf{x}_k | x_j |^2 \gamma^{(k)} (t, \mathbf{x}_k ; \mathbf{x}_k ).
\ee
Then we have

\begin{theorem}\label{th:VirialIdentityQuintic}
Assume that $\Gamma(t)= ( \gamma^{(k)}_t )_{k \geq 1}$ solves the Gross-Pitaevskii hierarchy \eqref{eq:GPHierarchyEquaFunctQuintic}. Then for any $k \ge 1,$
\beq\begin{split}\label{eq:VirialIdentityQuintic}
\partial^2_t V_k (\Gamma (t)) & = 8 k \int d \mathbf{p}_k |p_1|^2 \gamma^{(k)} (\mathbf{p}_k; \mathbf{p}_k) + \frac{8}{3} n k \mu \int d \mathbf{x}_k \gamma^{(k+1)} (\mathbf{x}_k, x_1; \mathbf{x}_k, x_1).
\end{split}\eeq
\end{theorem}

By Glassey's argument, we can conclude from Theorems \ref{th:k-energyconservationQuintic} and \ref{th:VirialIdentityQuintic} the following result.

\begin{theorem}\label{th:BlowupQuintic}
Let $n \ge 3.$ Assume that $\Gamma(t)= ( \gamma^{(k)}_t )_{k \geq 1}$ solves the focusing (i.e. $\mu=-1$) Gross-Pitaevskii hierarchy \eqref{eq:GPHierarchyEquaFunctQuintic} with initial condition $\Gamma(0) = ( \gamma^{(k)}_0 )_{k \geq 1} \in \mathfrak{H}^{1}.$ If $( \gamma^{(k)}_t )_{k \geq 1}$ is a sequence of density operators such that $V_k (\Gamma(0)) < \8$ and $E_k (\Gamma(0))<0$ for some $k \ge 1,$ then the solution $\Gamma(t)$ blows up in finite time with respect to $\mathfrak{H}^1.$
\end{theorem}

\subsection*{Acknowledgement} This research was supported in part by the NSFC under Grants No.11071095, 11171338, and 11101171.

\end{document}